\documentclass[journal,10pt]{IEEEtran}
\usepackage{amsthm,amsfonts}
\usepackage{amssymb}
\usepackage{algorithm}
\pdfoutput=1
\usepackage{algorithmic}
\usepackage{balance} 
\usepackage{cite,graphicx,amsmath,amssymb}
\usepackage{subfigure}
\usepackage{citesort}
\usepackage{fancyhdr}
\usepackage{tabularx}
\usepackage{mdwmath,mathtools}
\usepackage{mdwtab}
\usepackage{enumitem}
\usepackage{xcolor}
\usepackage{bm}
\usepackage{caption}
\usepackage{multirow}
\usepackage{flafter}
\usepackage{url}
\usepackage{longtable}
\usepackage{subfloat}
\usepackage{makecell}

\newtheorem{proposition}{Proposition}
\newtheorem{lemma}{Lemma}
\newtheorem{remark}{Remark}
\newcommand{\RomanNumeralCaps}[1]{\MakeUppercase{\romannumeral #1}}

\begin{document}
	\title{STARS-assisted Near-field ISAC: Sensor Deployment and Beamforming Design}
 
\author{ 
	Na~Xue,~\IEEEmembership{Graduate Student Member,~IEEE}, Xidong~Mu,~\IEEEmembership{Member,~IEEE},   Yue~Chen,~\IEEEmembership{Senior Member,~IEEE},
	Yuanwei~Liu,~\IEEEmembership{Fellow,~IEEE} 
	\thanks{Part of this paper has been accepted at the IEEE Global Communications Conference (GLOBECOM), Cape Town, South Africa, 8–12 Dec, 2023~\cite{Na_conf3}. } 
	\thanks{N. Xue, Y. Chen are with Queen Mary University of London, London, UK (email:\{n.xue,yue.chen\}@qmul.ac.uk). X. Mu is with Queen's University Belfast, Belfast, UK (email: x.mu@qub.ac.uk). Y. Liu are with the University of Hong Kong, Hong Kong (email: yuanwei@hku.hk). }  
	}
	\maketitle
	\vspace{-2.5cm}
	
	\begin{abstract} 
	A simultaneously transmitting and reflecting surface (STARS) assisted near-field (NF) integrated sensing and communication (ISAC) framework is proposed, where the radio sensors are installed on the STARS to directly conduct the distance-domain sensing by exploiting the characteristic spherical wavefront. A new squared position error bound (SPEB) expression is derived to reveal the dependence on beamforming (BF) design and sensor deployment. To balance the trade-off between the SPEB and the sensor deployment cost, a cost function minimization problem, a cost function minimization problem is formulated to jointly optimize the sensor deployment, the active and passive BF, subject to communication and power consumption constraints. For the sensor deployment optimization, a joint sensor deployment algorithm is proposed by invoking the successive convex approximation. Under a specific relationship between the sensor numbers and BF design, we derive the optimal sensor interval in a closed-form expression. For the joint BF optimization, a penalty-based method is invoked. Simulation results validated that the derived SPEB expression is close to the exact SPEB, which reveals the Fisher information Matrix of position estimation in NF can be approximated as a diagonal matrix. Furthermore, the proposed algorithms achieve the best SPEB performance than the benchmark schemes accompanying the lowest deployment cost.
	\end{abstract}
	
	
	\section{Introduction} 
	Wireless sensing has been deemed as an essential functionality for the sixth generation (6G) wireless networks to support the various emerging applications, such as extended reality (XR), Internet of Vehicles, and Metaverse~\cite{Saad20Network,Metaverse}. To fulfill the ubiquitous connectivity of the environment-aware applications, integrated sensing and communication (ISAC) utilizes the shared spectrum resources and infrastructures to unify the dual functionalities, where the electromagnetic wave simultaneously conveys the modulated information stream and conducts the wireless sensing (i.e., position, localization, tracking) via the reflected echoes~\cite{FanLiu18TWC,Fanliu_20TCOM}. Thus, the wireless propagation environment plays a vital role in the performance of ISAC. Due to the densely deployed infrastructure and obstacles (i.e. buildings and trees) in 6G wireless networks, it is hard to maintain the effective line-of-sight (LoS) links between the base station (BS) and users (i.e., communication user or sensing target), which inevitably degrades the performance of ISAC. 
	
    To overcome these limitations, reconfigurable intelligent surfaces (RISs) were proposed to reconstruct the wireless propagation environment by the phase shift control of the passive reflective elements~\cite{QingqingWu_21TCOM,QingqingWu_19TWC}. However, RISs only reflect the induced signal in one half-space without any coverage in another half-space. Fortunately, the simultaneously transmitting and reflecting RIS (STARS) were proposed to realize the full-space transmission~\cite{Yuanwei_21WC}. By adjusting the electric and magnetic currents of the STAR element, STARS split the induced signal into the reflecting signal and transmitting signal towards two sides to support full-space coverage. Considering the significant increase in the scale of RIS (STARS) and high frequencies in 6G, the near-field (NF) region has been extended to the range of hundreds of meters, where the electromagnetic (EM) waves are modelled as spherical wavefronts to precisely depict its characteristic. For a specific example of a single antenna at the BS, the cascaded links are under NF propagation as long as the transmitter–RIS distance or the RIS–receiver distance is less than $\frac{2 {D}_{{\rm{RIS}}}^{2}}{\lambda}$, where ${D}_{{\rm{RIS}}}$ and ${\lambda}$ denote the aperture size of RIS and the wavelength of the carrier frequency, respectively~\cite{MingyaoCui_Magzine,Xidong_NF_Mag}. For the massive antennas at the BS, the NF effect cannot be negligible in RIS-assisted 6G networks due to the further extended Rayleigh distance. Compared with the far-field (FF) channel, the direction and the distance of the signal can be analysed from the NF channel with the help of the antenna array~\cite{NF_yw}, which has the potential to elevate the ISAC performance.
    
	\subsection{Prior Work}
	
	 \begin{table*}[!ht]\large 
	 \caption{Our contributions in contrast to the state-of-the-art. }  \label{Table:Comparition}
		\begin{center}
			\centering
			\resizebox{\textwidth}{!}{
				\begin{tabular}{ !{\vrule width1.4pt}l !{\vrule width1.4pt}c !{\vrule width1.4pt}c!{\vrule width1.4pt}c!{\vrule width1.4pt}c!{\vrule width1.4pt}c!{\vrule width1.4pt}c!{\vrule width1.4pt} }
					\Xhline{1.4pt}
					\centering
					& ~\cite{XiaodanShan_22JSAC} & ~\cite{Yinghui_22JSAC} & ~\cite{Zhaolin_22STAR} & ~\cite{Adham_22TWC} & ~\cite{Haocheng_NF-ISAC} & \bf{Proposed} \\
					\Xhline{1.4pt}
					\centering 
					ISAC functionality & $\times$, radio sensing-only & $\surd$ & $\surd$ &$\times$ & $\surd$ & $\surd$ \\
					\hline
					\centering
					NF propagation & $\times$ & $\times$  &$\times$ & $\surd$ & $\surd$  &$\surd$ \\
					\hline
					\centering 
					Dedicated sensing beam & $\times$ & $\times$ & $\surd$ &$\surd$, sensing-only & $\times$ & $\surd$ \\
					\hline
					\centering
					RIS deployment & $\surd$ &$\surd$, separate RISs &$\surd$, STARS & $\times$  &$\times$  &$\surd$, STARS \\
					\hline
					\centering
					Sensor deployment & $\surd$,predefined & $\times$  &$\surd$, predefined & $\surd$, BS & $\times$ &$\surd$, ULA, optimized \\
					\hline
					\centering
					The radio sensing metric  & averaged received signal power & detection performance & {\rm{CRB}} & {\rm{CRB}} & {\rm{CRB}} & {\rm{SPEB}} \\
					\Xhline{1.4pt}
			\end{tabular}  }
		\end{center}
		\label{table:structure2}
	\end{table*}
	
	\subsubsection{STARS-assisted FF ISAC} 
     Benefiting from the characteristic of manufacturing the propagation environment, the interplay of RIS and ISAC has drawn great attention~\cite{XiaodanShan_22JSAC,Yinghui_22JSAC,XiaoMeng_22Conf}. By establishing a virtual link in the dead zone for both the communication users (CUs) and sensing targets (STs), the RIS-assisted ISAC provides an additional degree of freedom (DoF) to suppress the inter-functionality interference (i.e. the interference of radio sensing function on communication)~\cite{XiaoMeng_22Conf,Yinghui_22JSAC}. To further extend the coverage to support the ubiquitous connections, the STARS was proposed~\cite{JiaqiXu_21Letter,STAR_XD}. By manufacturing the transmission and reflection coefficients, the STAR elements work in energy splitting mode and are able to simultaneously support the reflecting and transmitting of the induced signal over different time slots. For example, a STARS-split ISAC framework is proposed, where the STs in the sensing half-space and the CUs in the communication half-space can be served at the same time~\cite{Zhaolin_22STAR}. To further support the CUs and STs among two sides of STARS, a cluster-based non-orthogonal multiple access (NOMA) assisted full-space STARS-ISAC framework is proposed~\cite{Na_24TWC}.

     \subsubsection{The Role of NF in ISAC}
     Benefiting from the embedded distance information in the phase of EM wavefronts, the research of NF has flourished in radio sensing and ISAC area~\cite{Liana_21TSP,Alessio_22JSTSP,Haocheng_3DLocalization,Yuhua_23TWC,Adham_22TWC}. For the radio sensing aspect, the conditional and unconditional Cram{\'e}r-Rao bound (CRB) under the NF propagation environment were derived in~\cite{Liana_21TSP}. The authors of~\cite{Alessio_22JSTSP} investigated the NF radio sensing based on a multiple-input-multiple-output (MIMO) testbed, which achieves a high degree of spatial resolution even with lower bandwidth. To fully exploit the embedded information in the NF channel modelling, the authors of~\cite{Haocheng_3DLocalization} analyzed the CRB of three dimension (3D) localization in a NF MIMO radar. To circumvent the blocked LoS link between the BS and the ST, a focal scanning method was adopted in a RIS-assisted system to roughly estimate the sensing target~\cite{Yuhua_23TWC}. By exploiting the additional angle information at the RIS, the authors of~\cite{Adham_22TWC} proposed a joint localization and synchronization algorithm under NF conditions to improve the accuracy of a single-antenna receiver. Furthermore, the NF propagation has also shown its advantages in beamfocusing to elevate communication performance~\cite{ZidongWu_23JSAC}. The authors in~\cite{Zhaolin23_NF} investigated the precoder design in NF ISAC, where the separate CRB for angle and distance was given.  
     
\subsection{Motivation and Contributions}
The motivations of our work can be summarised in the following three folds. 

\emph{Firstly}, conventional RIS-assisted ISAC cannot achieve a high estimation precision due to the multi-hop transmission links (i.e., BS$\to$STARS$\to$ST$\to$STARS$\to$BS). With the long round-trip distance during the multi-hop transmission links, the reflected echo signal at the BS induces high signal attenuation due to path loss. To combat this disadvantage, we employ the sensing-at-STARS architecture to elevate the radio sensing performance.

\emph{Secondly}, conventional FF ISAC in STARS-assisted systems cannot guarantee the estimation precision on distance. The multiple unknown parameters during the round-trip distance among the multi-hop transmission bring a great challenge to the precise estimation of distance. Benefiting from the spherical wavefronts modelling in NF propagation, the distance information can be directly estimated via the sensor arrays rather than the round-trip estimation. To the best of our knowledge, there is no existing work investigating the potential of NF-ISAC with the sensing-at-STARS architecture.   

\emph{Thirdly}, the densely embedded sensors at the STARS in NF ISAC system inevitably require high power consumption and hardware cost to work at a high frequency. Some prior works conduct qualitative analysis about the required number of sensor elements at STARS in FF~\cite{Zheng_24TWC,Zhaolin_22STAR}. How to evaluate the impact of sensor deployment on radio sensing performance and balance the radio sensing performance and the sensor deployment cost in a STARS-assisted NF-ISAC system is still open to be addressed. 

Motivated by these three points, we propose a STARS-assisted NF-ISAC framework, where the sensor elements are embedded at the STARS to simultaneously estimate the angle and distance. To evaluate the quantity impact of sensor deployment (i.e., sensor interval and the number of sensor elements), we first derive a new squared position error bound (SPEB). Inspired by the derived SPEB expression, the joint sensor deployment and beamforming (BF) design are investigated to balance the radio sensing performance and the sensor deployment cost. The main contributions of this work are summarized below. 

\begin{itemize} 
	\item We consider a STARS-assisted NF ISAC system, where the ISAC BS conduct the dual functionality via a STARS-enabled NF propagation link. To avoid signal attenuation among multi-hop sensing links, the `sensing-at-STARS' structure is employed. A new NF SPEB expression is derived to unveil the impact of sensor deployment on SPEB, which provides a guideline to investigate the sensor deployment and BF design problem to balance the SPEB performance and the sensor deployment cost. 

    \item We derive a new expression of NF SPEB by exploiting the property of the steering vector. Based on the derived SPEB expression, a weighted normalized SPEB and deployment cost minimization problem is proposed. For the sensor deployment optimization, both the sensor interval and the number of radio sensors are optimized to balance the SPEB performance and the sensor deployment cost. A joint successive convex approximation geometric programming algorithm is proposed to solve the mixed integer non-linear programming. By investigating the relationship between the number of sensors and BF design, a closed-form solution with a specific condition is given to reveal the aperture size compatibility among the transmitter and receiver.     

    \item We propose a penalty-based iterative algorithm to alternatively optimize the active BF at the BS and the passive BF at the STARS with the fixed sensor deployment. To deal with the rank one constraints of active BF, the Schur complement and semidefinite relaxation (SDR) approach is employed. While the successive convex approximation (SCA) is employed to deal with the non-convex constraints of passive BF at the STARS. 

    \item Our numerical results validate that the derived SPEB with the guideline of sensor deployment expression reaches the same performance compared with the state-of-art expression. Furthermore, the simulation results also verify the effectiveness of our proposed scheme. By the optimization of the sensor deployment at STARS, the performance of SPEB and weighted cost are both elevated. 
   
\end{itemize}
	
\subsection{Organization and Notation} 
\begin{table*}[!th]
	\caption{Summary of Key Symbols and their Definitions. } \label{table:notation}
	\begin{center}
		\begin{tabular}{|c|c|c|c|}
			\cline{1-4}
			{\textbf{Symbols}} & {\textbf{Definition}} ($^{*}$ denotes the optimized variable ) & {\textbf{Symbols}} & {\textbf{Definition}} ($^{*}$ denotes the optimized variable ) \\ \cline{1-4}
			$N$	&  Number of transmit antennas at the BS    &
			$M$   &  Number of the STAR elements    \\ \cline{1-4}
			${M}_{r}$ & Number of the sensor elements on STARS $^{*}$   & ${d}_{{\rm{B}}}$  & The antenna interval at the BS\\ \cline{1-4} 
			$d_{{\rm{R}}} $ & The interval between two adjacent STAR elements & 
			${d}_{s}$ & The interval between adjacent sensor elements  $^{*}$ \\ \cline{1-4} 
			${\mathbf{H}}_{{\rm{BR}}}$ & The LoS channel between the BS and the STARS & ${{\mathbf{g}}}_{c}$ & The channel between the STARS and the CU \\  \cline{1-4}
			${\bm{\alpha}}_{r}$ & The steering vector of ST $\to$ STARS & ${\bm{\alpha}}_{t} $ & The steering vector of STARS $\to$ ST \\ \cline{1-4}
			${\mathbf{w}}_{s}$ & The active BF vector for dedicated sensing signal & ${\mathbf{v}}_{c}$ & The active BF vector for joint C\&S beam \\ \cline{1-4}
			${\mathbf{x}}\left[n\right] $ & The transmit signal from the BS & ${\mathbf{x}}_{l} \left[n\right] $ & The reflected signal from the STARS \\ \cline{1-4} 
			${\mathbf{A}} \left( {r}_{s}, {\theta}_{s} \right) $  & The correlation matrix of the steering vector &
			${\dot{{\mathbf{A}}}}_{x}, x \in \left\{{r}_{s}, {\theta}_{s}\right\} $ & The partial derivatives of ${\mathbf{A}}$ w.r.t. $x$  \\ \cline{1-4}
			${\mathbf{v}}_{r,1} ({\mathbf{v}}_{t,1} ) $ & The auxiliary vector of the sensor (STAR element) for $ \frac{ \partial{{{\bm{\alpha}}}_{r}} } {{\partial{{r}_{s}}}}$ & ${\bm{\eta}} $ & The Cardesian coordinate of ST   \\ \cline{1-4}
			${\mathbf{v}}_{r,2} ({\mathbf{v}}_{t,2} ) $ &  The auxiliary vector of the sensor (STAR element) for $\frac{ \partial{{{\bm{\alpha}}}_{r}} } {{\partial{{\theta}_{s}}}}$ & ${\bm{\gamma}}$ & The polar coordinate of ST \\ \cline{1-4} 
			${{\bm{\Theta}}}_{l}$ & The diagonal coefficient matrix of STARS for $l, l \in \left\{r,t\right\}  $ space  & ${\mathbf{R}}_{x} $
			& Covariance matrix of ${\mathbf{x}}\left[n\right] $ at the BS \\ \cline{1-4}
			${\mathbf{R}}_{{\bar{{\mathbf{X}}}}_{r}} $ & Covariance matrix of ${\mathbf{x}}_{r}\left[n\right] $ at the STARS  & ${\mathbf{Y}}_{s}$
			& received echo signal at the sensors \\ \cline{1-4}   
		\end{tabular}
	\end{center}
\end{table*}

The rest of this paper is organized as follows. In Section \RomanNumeralCaps{2}, a STARS-ISAC system under NF propagation is presented, where the approximate NF SPEB expression is given. A weighted SPEB and deployment cost minimization problem is formulated by jointly optimizing the sensor deployment, active and passive BF design. In Section \RomanNumeralCaps{3}, the joint sensor deployment algorithm is first developed for the sensor deployment problem. After the decoupling of active and passive BF, the penalty-based BF design algorithm is developed. In Section \RomanNumeralCaps{4}, numerical results are provided to validate the precision of the derived SPEB expression and evaluate the characteristics of our proposed algorithms. Finally, we conclude this paper in Section \RomanNumeralCaps{5}. Table~\ref{table:notation} summarises the key parameters and their definitions throughout the whole paper.\\

\emph{Notations:} Scalars, vectors, and matrices are denoted by lower-case, bold-face lower-case, and bold-face upper-case letters, respectively; ${\mathbb{C}}^{N \times 1}$ denotes the $N \times 1$ complex-valued vectors; Symbol $\odot$ denotes the Hadamard (element-wise) product between two matrices (vectors);  $x^{*}$ denotes the complex conjugate of $x$; ${\mathbf{a}}^{H}$, $\left\| {\mathbf{a}} \right\|_{l}$ represent the conjugate transpose, the $l, l \in \left\{1,2\right\}$ norm of vector $\mathbf{a}$; Rank(${\mathbf{A}}$) and Tr(${\mathbf{A}}$) denote the rank and trace of matrix ${\mathbf{A}}$, respectively. ${\mathbf{A}} \succeq 0$ means that ${\mathbf{A}}$ is a positive semidefinite matrix. 
	
\section{System Model and Problem Formulation}
	\begin{figure}[t]
		\centering
		\includegraphics[width=3.5in]{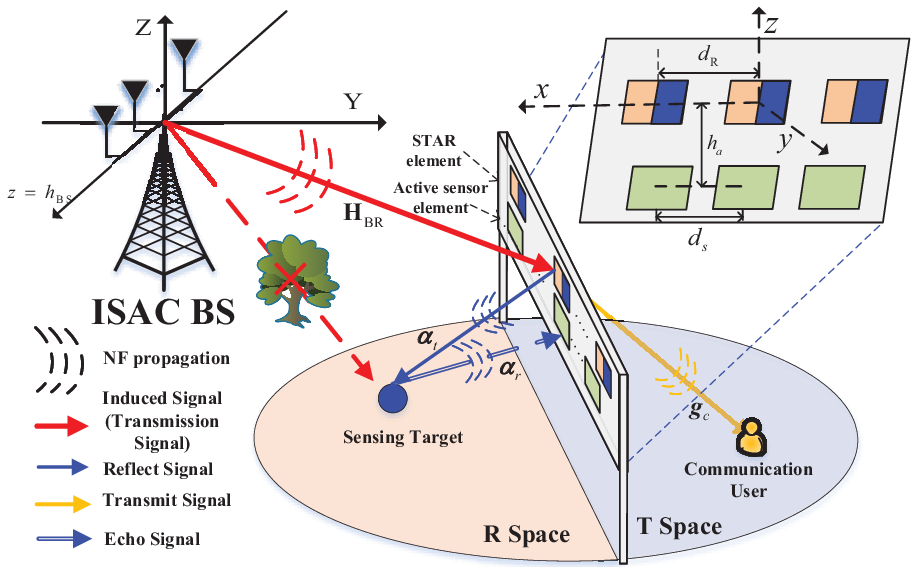}\\
		\caption{The STARS-assisted NF-ISAC System.}\label{Fig:System_NearField_ISAC_STAR}
	\end{figure}
	A STARS-assisted NF ISAC system is shown in Fig.~\ref{Fig:System_NearField_ISAC_STAR}, where an ISAC BS is equipped with $N$-antenna uniform linear array (ULA). The STARS is composed with STAR elements, sensor elements and a central controller, where the STAR elements are explored for the transmission, the sensor elements are explored for radio sensing functionality, and the central controller are equipped to monitor the status of the other elements. The single-antenna communication (CU) is located in the T space while the sensing target (ST) is located in the R space\footnote{It is noted that the proposed framework and algorithm can be extended to the scenario where STs are located in T space.}. To further illustrate the system model, we successively introduce the STARS model, NF propagation model and ISAC model in the following parts. 
	
	\subsection{STARS Model}
	By manipulating the electric and magnetic current on the surface, the STAR elements are exploited to split the induced signal into the reflecting and transmitting signals. The STARS is equipped with $M$ uniform linear distributed STAR elements with the equal interval ${d}_{{\rm{R}}}$. The STAR elements work in the energy split mode, which simultaneously transmits and reflects the induced signal. Let ${\beta}_{l,m},  l \in \left\{r,t\right\}, \forall m \in {\mathcal{M}}_{R}=\left\{-{\tilde{M}}, -{\tilde{M}}+1, ..., 0, ... , {\tilde{M}}-1, {\tilde{M}} \right\}, {\tilde{M}}=\frac{M}{2} $ denote the amplitude coefficients towards the reflecting and transmitting signal, respectively. The amplitude coefficients of $m$th STAR element are constrained by the energy conservation law ${\beta}_{r,m}^{2}+{\beta}_{t,m}^{2} \leq 1$. Let $x_{m}$ denote the induced signal on the $m$th element of STAR, the transmitting signal $t_{m}$ and the reflecting signal $r_{m}$ across STAR elements can be expressed as
	\begin{equation}
	   t_{m}=\left( {\beta}_{t,m} e^{j {\theta}_{t,m}} \right) x_{m}, \; 
		r_{m}=\left( {\beta}_{r,m} e^{j {\theta}_{r,m}} \right) x_{m}, 
	\end{equation}where $\theta_{t,m}, \theta_{r,m} \in \left[0,2\pi \right]$ denote the phase shift of the $m$th STAR element towards the transmitting signal and reflecting signal. To depict the whole transmitting and reflecting response of STARS, we introduce the transmitting coefficients matrix ${\mathbf{\Theta}}_{t}$ and the reflecting coefficients matrix ${\mathbf{\Theta}}_{r}$, which can be expressed as ${\mathbf{\Theta}}_{l}=\operatorname{diag}\left(\beta_{l, 1} e^{\jmath \theta_{l, 1}}, \ldots, \beta_{l, M} e^{\jmath \theta_{l, M}}\right) \in {\mathbb{C}}^{M \times M}, l \in\{t, r\} $. For the radio sensing functionality, the sensors are co-located at the STARS to collect the reflected echo signal to avoid the signal attenuation of the multi-hop transmissions (i.e., BS $\to$ STARS $\to$ ST $\to$ STARS $\to$ BS). Without loss of generality, we assume that the sensor elements keep the same distance ${h}_{a}$ meter (m) in the z-axis from the STAR elements along the uniform linear positions. Besides the reflected echo signal from the ST, the sensors also receive the transmit signal from the BS and the other STAR element. By exploiting the signal process of the sensors, this interference can be eliminated via offline training~\cite{XiaodanShan_22JSAC} or successive interference cancellation, where we assume there is no interference at the sensors to investigate the sensor deployment impact on radio sensing performance.
	
	\subsection{Near-field Propagation Models}
	Due to the high frequency and massive antenna and STAR elements, both the CU and ST are within the NF region among STARS. As the phase of EM wavefronts in NF simultaneously depends on the angle and distance, we successively illustrate the NF channel model among the STARS $\to$ CU/ST and the NF channel model among the BS $\to$ STARS. 
    \subsubsection{The NF propagation link between the BS and the STARS}
    As shown in Fig.~\ref{Fig:System_NearField_ISAC_STAR}, the BS is equipped with $N$ ULA antennas along $y=-r_{{\rm{BR}}} $, where ${r}_{{\rm{BR}}}$ denotes the distance between the BS and the STARS. We take the central antenna in BS ULA as a reference point with the coordinate ${\mathbf{p}}_{{\rm{BS}},0}=\left[0, -r_{{\rm{BR}}} \right]^{{\rm{T}}}$, where $r_{{\rm{BR}}}$ denotes the horizontal distance between the BS and the STARS in the y-axis. With the antenna interval $d_{{\rm{B}}}$, the coordinate of the $i$th transmit antenna is ${\mathbf{p}}_{{\rm{B}}}^{i}=\left( i d_{{\rm{B}}}, -r_{{\rm{BR}}}\right), i \in \left\{ -\frac{N}{2},-\frac{N}{2}+1, ..., \frac{N}{2}-1, \frac{N}{2} \right\} $. The distance between $i$th transmit antenna and $m$th STAR element is given by:
    \begin{equation}
    	r_{{\rm{BR}}}\left(m,i\right)= \left\|{ {\mathbf{p}}_{{\rm{B}}}^{i}- {\mathbf{p}}_{\rm{R}}^{m} }\right\|_{2}.
    \end{equation}The line-of-sight (LoS) channel $ {\mathbf{H}}_{{\rm{BR}}} \in {\mathbb{C}}^{M \times N}$ between the BS and STARS can be denoted as
    \begin{equation}
    	\left[ {\mathbf{H}}_{{\rm{BR}}} \right]_{ \left(m,i\right) }=\frac{ {\lambda}_{c} }{ 4\pi r_{{\rm{BR}}}\left(m,i\right) } e^{-j \left(r_{{\rm{BR}}}\left(m,i\right) -r_{{\rm{BR}}} \right)}.
    \end{equation}
    
    \begin{remark}
    	Compared with FF radio sensing, the NF radio sensing elevates radio sensing performance via the joint angular and distance estimation. The conventional FF radio sensing conducts the distance estimation with the help of wideband subcarriers~\cite{Alessio_22JSTSP}, which consumes more available spectrum resources. By exploiting the involved distance information in the NF steering vector, the distance ${r}_{s}$ from the STARS central can be directly estimated with the narrow band. Furthermore, the NF channel also enables more precise beam focusing, which alleviates the signal leakage to the undesired area to avoid the interference to the echo signals.
    \end{remark}
    
    \subsubsection{The NF propagation link between the STARS and ST/CU}
    As illustrated in Fig.~\ref{Fig:System_NearField_ISAC_STAR}, the STARS is equipped with $M$ uniformly linear deployment STAR elements in the $x-z$ plane. Without loss of generality, the central STAR element is chosen as the reference point with the coordinate ${\mathbf{p}}_{ {\rm{STAR}},0 }=\left[ 0,0 \right]^{{\rm{T}}}$. The coordinate of the CU in T space is given by ${\mathbf{p}}_{{\rm{CU}}}=\left[-{r}_{c} \cos \left({\theta}_{c}\right), -{r}_{c}\sin \left({\theta}_{c}\right) \right]^{{\rm{T}}}$ while the coordinate of the ST in R space is given by $\left[{r}_{s} \cos\left({\theta}_{s}\right), {r}_{s} \sin \left({\theta}_{s}\right) \right]^{{\rm{T}}} $, where ${r}_{z} \in \left(0, {r}_{{\rm{Ray}}}\right), z \in \left\{c,s\right\} $ denotes the distance between the STARS to the CU (ST) and ${\theta}_{z} \in \left( -\frac{\pi}{2}, \frac{\pi}{2} \right), z \in \left\{c,s\right\}$ denotes the azimuth angle between the transmit sides (reflect sides) of the STARS and the CU (ST). The adjacent STAR elements are equally spaced with $d_{\rm{R}}$. The propagation distance between the STAR element $m$ and the CU (ST) is given by: 
\begin{equation}
	\begin{split}
		\left\| { {\mathbf{r}}_{z}-{\mathbf{p}}_{\rm{R}}^{m} } \right\|_{2} =
		{r}_{z}^{ \left(m\right) }\left( r_{z},{\theta}_{z} \right)  \mathop  \approx \limits^{(a)} {r}_{z}-md_{{\rm{R}}}\cos{\theta_{z} }+\frac{ m^{2}d_{{\rm{R}}}^{2} } {2{r}_{z}}, 
	\end{split}
\end{equation}where $\mathop  \approx \limits^{(a)} $ comes from the Taylor expression $ \sqrt {1 + x} \approx 1 + \frac{1}{2}x + O({x^2})$.    
Considering the high frequency, the massive antennas and STAR elements array, the electromagnetic wave is depicted by the spherical wavefronts. The spherical wave propagation induces different amplitude and phase shifts among different channel links between each STAR element and the ST/CU. Thus, the corresponding steering vector ${\bm{\alpha}}\left( r_{z}, {\theta}_{z} \right) \in {\mathbb{C}}^{{M} \times 1 }, z \in \left\{ c,s \right\} $ between STAR element $m$ and the CU (ST) is given by: 
\begin{equation}\label{NearField_Steer}
	\left[{\bm{\alpha}} \left( r_{z},{\theta}_{z} \right)\right]_{m} 
		= e^{ -j\frac{2\pi}{ {\lambda}_{c} } \left( -m d_{\rm{R}}\cos{\theta_{z}}+\frac{ m^{2}d_{\rm{R}}^{2} } {2{r}_{z}} \right) }. 
\end{equation}Based on the steering vector expression in \eqref{NearField_Steer}, the communication channel $ {\mathbf{g}}_{c} \in {\mathbb{C}}^{M \times 1 }$ between STARS and the CU is given by 
\begin{equation}
	{\mathbf{g}}_{c}= {\bm{{\beta}}}\left( r_{z} \right) \odot  {\bm{\alpha}} \left( r_{z},{\theta}_{z} \right),
\end{equation} where ${\bm{{\beta}}}\left( r_{z} \right) \in {\mathbb{C}}^{ {M} \times 1 }$ denotes the complex channel gain with the expression $\left[ {\bm{{\beta}}}\left( r_{z} \right) \right]_{m} = \frac{ {\lambda}_{c} }{4\pi r_{z}^{m} \left( r_{z},{\theta}_{z} \right)}$~\cite{MingyaoCui_23TWC}. To investigate the lower bound of the radio sensing performance, we assume that the CSI of all channels (i.e., ${\mathbf{H}}_{{\rm{BR}}}$ and ${\mathbf{g}}_{c} $) is perfectly known at the BS, where the CSI can be estimated via the numerous efficient channel estimation schemes~\cite{NF_CSI_23TVT,NF_CSI_24JSAC}.

\subsubsection{The NF propagation link between the ST and the sensors at STARS}
As illustrated in Fig.~\ref{Fig:System_NearField_ISAC_STAR}, the sensor array at the STARS is composed of $M_{r}$ uniformly linear deployment sensors in the $x-z$ plane. The adjacent sensors are equally spaced with $d_{\rm{s}}$. The corresponding coordinates of the $m$th sensors $m, \forall m \in {\mathcal{M}}_{s}=\left\{-{\tilde{M}}_{r}, -{\tilde{M}}_{r}+1, ..., {\tilde{M}}_{r}-1,  {\tilde{M}}_{r} \right\} $ is given by ${\mathbf{p}}_{\rm{s}}^{m}=\left[md_{{\rm{s}}}, -{h}_{a}\right]^{T} $, where ${h}_{a}$ denotes the distance in along z axis between the STAR array and the sensor array. Similar to the NF propagation link between the STAR and the ST, the propagation distance between the ST and the sensor element $m$ can be expressed as $\left\| { {\mathbf{r}}_{s}-{\mathbf{p}}_{\rm{s}}^{m} } \right\|_{2} \approx {r}_{s}-md_{{\rm{s}}}\cos{\theta_{s} }+\frac{ m^{2}d_{{\rm{s}}}^{2} } {2{r}_{s}}$. As the ST is located in the NF region, the corresponding steering vector ${\bm{\alpha}}_{r}\left( {\tilde{r}}_{s}, {\theta}_{s} \right) \in {\mathbb{C}}^{{M} \times 1 }$ between the ST and sensor element $m$ is given by: 
	\begin{equation}\label{NF_STtoSTARS}
		\left[{\bm{\alpha}}_{r} \left( {\tilde{r}}_{s},{\theta}_{s} \right)\right]_{m} 
		= e^{ -j\frac{2\pi}{ {\lambda}_{c} } \left( -m d_{\rm{s}}\cos{\theta_{s}}+\frac{ m^{2}d_{\rm{s}}^{2} } {2{\tilde{r}}_{s}} \right) }. 
\end{equation} 

\vspace{-0.6cm}
	\subsection{ISAC Model}
	
    \subsubsection{Communication Model}
    For the communication functionality, the joint communication and sensing (C\&S) beam is exploited to convey the modulated information while the dedicated sensing beam is utilized to improve the radio sensing performance. Let $s_{1}\left[ n \right], c_{1}\left[ n \right] $ denote the dedicated sensing signal towards the ST and the joint communication and sensing (C\&S) signal towards CU, respectively. Then, the transmit signal from BS can be expressed as: 
    \begin{equation}\label{eq:transmit signal}
    	{\mathbf{x}}[n]={\mathbf{x}}_{sen}[n]+ {\mathbf{x}}_{{\rm{ISAC}}}[n] = {\mathbf{w}}_{s}s_{1}\left[ n \right] + {{\mathbf{v}}_c} c_{1}\left[ n \right], 
    \end{equation}where ${\mathbf{w}}_{s} \in {\mathbb{C}}^{N \times 1}$ and ${\mathbf{v}}_{c} \in {\mathbb{C}}^{N \times 1}$ denote the active BF vector for dedicated sensing signal and the active BF vector for CU, respectively. Without loss of generality, we assume that the dedicated sensing signal $s_{1}\left[ n \right] $ and joint C\&S signal $c_{1}\left[ n \right] $ are statistically independent with each other and all of them have zero mean and unit power. Thus, the total transmission power consumption for the integrated signal is given by \begin{equation}
    	P={\mathbb{E}}\left[ {{{\mathbf{x}}^H}\left[ n \right] {\mathbf{x}} \left[ n \right] } \right] = {{\mathbf{w}}_s^H} {{\mathbf{w}}_s} + {{\mathbf{v}}_c^H} {{\mathbf{v}}_c}. 
    \end{equation}
    The corresponding covariance matrix of the transmit signal is given by $ {\mathbf{R}}_{x}= {\mathbf{w}}_{s} {\mathbf{w}}_{s}^{H} +  {{\mathbf{v}}_c} {\mathbf{v}}_{c}^{H}$. The received signal at the CU can be expressed as 
    \begin{equation}
    	\begin{split}
    		&y_{c} \left[n\right]= {\mathbf{g}}_{c}^{H} {\mathbf{\Theta}}_{t} {\mathbf{H}}_{{\rm{BR}}} \left( {\mathbf{w}}_{s} s_{1}\left[n\right] + {{\mathbf{v}}_{c}} c_{1}\left[n\right] \right) + n_{0},\\
    	\end{split}
    \end{equation}where $n_{0} \sim {\mathcal{CN}} \left(0, \sigma^{2}\right)$ is the additive white Gaussian noise (AWGN). The achievable data rate of CU can be expressed as 
    \begin{equation}
    	\begin{split}
    		R_{1}=& {{\log }_{2}} ( 1+ \frac{ { {\left| {\mathbf{g}}_{c}^{H}{\mathbf{\Theta}}_{t}{\mathbf{H}}_{{\rm{BR}}} {\mathbf{v}}_{c} \right|}^2} } { \left| {\mathbf{g}}_{c}^{H}{\mathbf{\Theta}}_{t}{\mathbf{H}}_{{\rm{BR}}} {\mathbf{w}}_{s} \right|^{2}  +\sigma^{2}} ).
    	\end{split}
    \end{equation}
    
	\subsubsection{Sensing Model}
    For radio sensing functionality, both the joint C\&S beam accompanied with the dedicated sensing beam are exploited. As the CU and ST is located in different space among the STARS, the joint C\&S beam is not sufficient to guarantee the precise requirement of ST. Thus, the dedicated sensing beam is utilized to further improve the radio sensing performance. The transmit signal at the BS is given by \eqref{eq:transmit signal}. In practice, the covariance matrix ${\mathbf{R}}_{x} \in {\mathbb{C}}^{N \times N}$ of the transmit signal from the BS can be calculated by
    \begin{equation}
    	{\mathbf{R}}_{x}={\frac{1}{L}}  \left( \bar{{\mathbf{X}}} \bar{{\mathbf{X}}}^{H} \right),
    \end{equation}where $\bar{{\mathbf{X}}}=\left[ {\bar{{\mathbf{x}}}} \left[1\right], ..., {\bar{{\mathbf{x}}}} \left[L\right] \right] \in {\mathbb{C}}^{N \times L}$ denotes the observed transmit signal during one coherent block on $L$ slots. As the LoS link between the BS and the ST is blocked by the obstacle, the reflected echo signal $ {\mathbf{x}}_{r} \left[n\right]  \in {\mathbb{C}}^{M \times 1} $ that conducts the radio sensing functionality at the STARS is composed by the transmission link BS $\to$ STARS $\to$ ST, which is given by ${\mathbf{x}}_{r} \left[n\right]={\bm{\Theta}}_{r} {\mathbf{H}}_{{\rm{BR}}} {\mathbf{x}}\left[n\right]+{n}_{0}$. 
    The received echo signal from the ST at the STARS can be expressed as 
    \begin{equation}\label{Received_Echo_Signal}
    	{\mathbf{Y}}_{s}={\alpha}_{s} {\bm{\alpha}}_{r} \left( {\tilde{r}}_{s}, {\theta}_{s}\right) {\bm{\alpha}}^{H}_{t} \left( r_{s}, {\theta}_{s} \right){\mathbf{\Theta}}_{r} {\mathbf{H}}_{\rm{BR}} {\mathbf{R}}_{x}+{\mathbf{N}}_{s},
    \end{equation}where ${\tilde{r}}=\sqrt{ r^{2}-{h}_{a}^{2} }$ denotes the distance from the ST to the central of sensor elements. ${\alpha}_{s} \in {\mathbb{C}} $ denotes the reflection coefficient which depends on both the round-trip path-loss and the radar cross section (RCS) of the target.

    \subsubsection{Radio Sensing Metric: Squared Position error bound (SPEB)}
    By exploiting the spherical EM wavefronts in the NF propagation, the distance $r_{s}$ and angle ${\theta}_{s}$ can be separately estimated~\cite{Zhaolin23_NF}, the joint positioning performance is still hard to be evaluated from the separate estimation. Driven by this, the SPEB, which denotes the performance bound of the unbiased position estimator, is adopted to evaluate the position performance in the STARS-assisted NF-ISAC system. Instead of carrying out the separate estimation on angle or distance with the CRB metric, the SPEB evaluates the fundamental limits of the positioning accuracy of the unknown position vector ${\bm{{\eta}}}$. For example, let ${\bm{\eta}} \buildrel \Delta \over=\left[p_{x} \; p_{y} \right]^{{\rm{T}}} \in {{\mathbb{R}}}^{2 \times 1}$ denotes the unknown position vector in the Cartesian coordinate, where ${p}_{x} $ and ${p}_{y}$ denotes the coordinate in the $x$-axis and $y$-axis, respectively. The SPEB of ${\bm{\eta}}$ can be calculated by 	
    \begin{equation}\label{SPEB_def}
    	{\rm{SPEB}} \left(\{ {\mathbf{R}}_{x}, {\mathbf{q}}_{r}, {\bm{\Theta}}_{r}\} ; {\bm{\eta}} \right) \buildrel \Delta \over = {\rm{tr}} \left( {\mathbf{J}}_{\bm{\eta}}^{-1} \right)\leq  {\mathbb{E}}\big\{ \left\|{ {\hat{ {\bm{\eta}} }} - {\bm{\eta}} }\right\|^2  \big\} ,
    \end{equation}where ${\bm{J}}_{{\bm{\eta}}}$ denotes the Fisher information Matrix (FIM) of ${\bm{\eta}}$ and ${{\hat{ {\bm{\eta}} }}} $ denotes the unbiased estimator of ${\bm{\eta}}$. Due to the NF propagation, the elements in ${\bm{\eta}} $ are quite difficult to directly estimated from the echo signal while the amzith angle and distance in the polar coordinate are more easy to be estimated. By denoting the polar coordinate of ST as ${\bm{\gamma}} \buildrel \Delta \over =  \left[ {r}_{s} \; {\theta}_{s} \right]^{{\rm{T}}} $, the matrix transformation $ {\bm{T}} \buildrel \Delta \over =  \frac{\partial \bm{\gamma}^{\mathsf{H}}} {\partial \bm{\eta}} \in {\mathbb{R}}^{2 \times 2} $ can be exploited to map the polar coordinate ${\bm{\gamma}}$ into the Cartesian coordinate ${\bm{\eta}}$ based on its geometric relationship. Therefore, the ${\bm{J}}_{{\bm{\eta}}}^{-1} \in {\mathbb{R}}^{2 \times 2}$ can be calculated via ${\bm{J}}_{{\bm{\eta}}}^{-1} =\left( {\mathbf{T}}^{H} \right)^{-1} {\bm{J}}_{{\bm{\gamma}}}^{-1} {\mathbf{T}}^{-1}  $. ${\bm{J}}_{{\bm{\gamma}}}^{-1} $ denotes the FIM of ${\bm{\gamma}}$, which satisfies ${\mathbb{E}} \left\{ ( {\hat{{\bm{\gamma}}}} - {\bm{{\gamma}}}) ( {\hat{{\bm{\gamma}}}} - {{\bm{\gamma}}})^{{\rm{T}}} \right\} -  {\mathbf{J}}_{\bm{\gamma}}^{-1} \succeq 0 $. To calculate ${\bm{J}}_{\eta}^{-1}$, we first vectorise the received echo signal of \eqref{Received_Echo_Signal}, which can be derived as 
    	\begin{equation}\label{Received_Echo_Signal_vec}
    		{\mathbf{y}}_{s} ={\rm{vec}}\left( {\mathbf{Y}}_{s} \right)= {\mathbf{m}}_{s}+{\mathbf{n}}_{s},
    	\end{equation}where ${\mathbf{m}}_{s}={\rm{vec}} \left( {\alpha}_{s} {\bm{\alpha}}_{r} \left( {\tilde{r}}_{s},  {\theta}_{s} \right) {\bm{\alpha}}^{H}_{t} \left(  r_{s}, {\theta}_{s} \right) {\bm{\Theta}}_{r} {\mathbf{H}}_{{\rm{BR}}} {\bar{{\mathbf{X}}}} \right) $ is the desired part in the received signal ${\mathbf{y}}_{s}$ at ST and ${\mathbf{n}}_{s}={\rm{vec}} \left( {\mathbf{N}}_{s}\right) $. As the observation in \eqref{Received_Echo_Signal_vec} follows the complex Gaussian distribution, the $\left(i,j\right)$-th FIM elements $\left[ {\mathbf{J}}_{\bm{\gamma}} \right]_{i,j} \buildrel \Delta \over = \Lambda(\gamma_i,\gamma_j) $ can be expressed as~\cite{kay1993fundamentals}   \begin{equation}\label{FIM_element0}
    		\Lambda({\gamma}_{i}, {\gamma}_{j}) = \frac{2}{\sigma^2} \sum_{n=0}^{L-1} \Re\left\{\left( \frac{\partial m_{s}[n]} {\partial \gamma_i} \right)^{*} \frac{\partial m_{s}[n]} {\partial \gamma_j}\right\}. 
    	\end{equation}Let ${\mathbf{A}}\left( {r_s}, {\theta_s} \right)\buildrel \Delta \over = {{\bm{\alpha}}_r}\left( {{{\theta}}_s, {{\tilde{r}}_s}} \right) {\bm{\alpha}} _t^H\left( {{\theta_s},{r_s}} \right) $, the partial derivative of ${\mathbf{m}}_{s}$ w.r.t. $\left\{ {r}_{s}, \theta_{s} \right\}$ can be expressed as\footnote{To simplify the expression, we utilize ${\bm{\alpha}}_{l}, l \in \left\{r,t\right\}$ to denote ${\bm{\alpha}}_{l}\left( {r_z}, {\theta_z} \right) $ in the following derivation.} $\frac{{\partial {{\mathbf{m}}_s}}} {{\partial {r_s}}} = {\rm{vec}}\left( {\alpha _s} {\dot {\mathbf{A}}}_{{r_s}}\left( {r_s}, {\theta_s} \right) {{\bm{\Theta}}_r}{{\mathbf{H}}_{{\rm{BR}}}} {\bar{{\mathbf{X}}}} \right)$, and $\frac{ {\partial {\mathbf{m}}_{s} } }{ {\partial {\theta_{s}}} }= {\rm{vec}}\left( { {{\alpha}_s} {{\dot {\mathbf{A}}}_{{\theta_s}}} \left( {r_s}, {\theta_s}  \right)  {{\bm{\Theta}}_r} {{\mathbf{H}}_{{\rm{BR}}}} {\bar{{\mathbf{X}}}} } \right)$. $ {{\dot {\mathbf{A}}}_{{r_s}}} \left( {r_s}, {\theta_s} \right) $ and ${\dot {\mathbf{A}}}_{{{\theta}_s}} \left( {r_s}, {\theta_s} \right) $ denote the partial derivatives of ${\mathbf{A}} \left(r_{s}, {\theta}_{s} \right) $ on ${r}_{s}$ and ${\theta}_{s} $, respectively. Due to the symmetry structure of $\left\{ {\bm{\alpha}}_{t}, {\bm{\alpha}}_{r} \right\}$, it is clear that ${\bm{\alpha}}_{t}^{H} {{\dot {\bm{\alpha}} }_{t, i}} =  {\bm{\alpha}}_r^H {{\dot {\bm{\alpha}} }_{r,i}}= 0, i \in {1,2}$, where ${\bm{\alpha}}_l, l \in \left\{r,t\right\}$ denotes the partial derivation of ${\bm{\alpha}}_l$ w.r.t. $\left\{ {\theta}_{s}, {r}_{s}\right\} $ with the expression $   \left[\dot{\bm{\alpha}}_{l,{\theta _s}}\right]_{m}=\frac{{\partial {{\left[ {{{\bm{\alpha}} _l}} \right]}_{m}}}} {{\partial {\theta_s}}} = \frac{{j 2\pi md_{l}\sin {\theta_s}}} {{{\lambda _c}}} \left[ {{{\bm{\alpha }}_l}} \right]_m $ and $\left[\dot{\bm{\alpha}}_{l,{r_s}}\right]_{m}=\frac{{\partial {{\left[ {{{\bm{\alpha}}_l}} \right]}_{m}}}} {{\partial {r_l}}}=\frac{{j \pi {m^2}{d_{l}^2} }} {{{\lambda_c} {r}_l^2}} \left[ {{{\bm{\alpha }}_l}} \right]_m, l \in \left\{r,t\right\}, d_{r}=d_{s}, d_{t}=d_{{\rm{R}}}, r_{r}={\tilde{r}}_{s}, r_{t}=r_{s} $. To unveil the impact of sensor deployment on SPEB, we introduce the integer vector $ {{\mathbf{v}}_{l,i}}, l \in \left\{r,t\right\}, i \in \left\{1,2\right\} $, which is given by
    \begin{subequations}\label{vector_number}
    	\begin{equation}
    		{{\mathbf{v}}_{l,1}} = [ - {\tilde{M}}_{l}, - {\tilde{M}}_{l} + 1,..., {\tilde{M}}_{l} - 1,{\tilde{M}}_{l}] \in {{\mathbb{R}}^{M_{l} \times 1}},
    	\end{equation}
    	\begin{equation}
    		{\mathbf{v}}_{l,2} = [ {{\tilde{M}}_{l}^2},{( - {\tilde{M}}_{l} + 1)^2},...,{({\tilde{M}}_{l} - 1)^2},{{\tilde M}_{l}^2}] \in {{\mathbb{R}}^{M_{l} \times 1}}.
    	\end{equation}
    \end{subequations}The partial derivation of ${\bm{\alpha}}_l, l \in \left\{r,t\right\} $ can be expressed as $\frac{{\partial {{\left[ {{{\bm{\alpha}}_l}} \right]}_{m}}}} {{\partial {\theta_s}}} = \frac{{j2\pi d_{l} \sin {\theta_s}}} {{{\lambda_c}}} {{\mathbf{v}}_{l,1}} \odot {{\bm{\alpha}}_{l}} $, and $\frac{{\partial {{\left[ {{{\bm{\alpha}}_l}} \right]}_{m}}}} {{\partial {r_l}}}=
\frac{ {j\pi d_{l}^{2} }} { {{\lambda_c}} {r}_{l}^{2} } {\mathbf{v}}_{l,2} \odot {{{\bm{\alpha}}_l}}$. According to the definition of Hadmard product, $ {{\mathbf{v}}_r} \odot {{\bm{\alpha}}_r}={\rm{diag}} \left({\mathbf{v}}_{r}\right) {\bm{\alpha}}_{r}$ and ${\mathbf{v}}_{r} \odot {\bm{\alpha}}_{r} {\bm{\alpha}}_{t}^{H}$ can be further derived as follows
\begin{equation}\label{App3_Pro}
	{\mathbf{v}}_{r} \odot {\bm{\alpha}}_{r} {\bm{\alpha}}_{t}^{H}={\rm{diag}}({\mathbf{v}}_{r}) {\mathbf{A}} \in {\mathbb{C}}^{{M}_{r} \times M}. 
\end{equation} Thus, the partial derivation ${{\dot {\mathbf{A}}}_{{\theta_s}}}, {\dot {\mathbf{A}}} _{{r_s}}=\frac{\partial {\mathbf{A}} \left({\theta_s},{r_s} \right)} {\partial {r_s}} \in {\mathbb{C}}^{{M_t} \times {M_r}}$ can be expressed as\begin{subequations}\label{PartialDerivaation_A_ori}
	\begin{equation}
		{{\dot {\mathbf{A}}}_{{\theta_s}}} \left({\theta _s},{r_s} \right)=\frac{ {j2\pi \sin {\theta_s}} } {{{\lambda_c}}}[d_{s}{\rm{diag}}( {{\mathbf{v}}_{r,1}} ) {\mathbf{A}} - {d}_{{\rm{R}}} {\mathbf{A}} {\rm{diag}}({{\mathbf{v}}_{t,1}})],
	\end{equation}
	\begin{equation}
		{{\dot {\mathbf{A}}}_{{r_s}}} \left( {\theta _s},{r_s} \right) = \frac{{j\pi }} {{{\lambda _c}}} [ \frac{ {{d_{s}^2}}} { {\tilde{r}}_s^2 }{\rm{diag}}({{\mathbf{v}}_{r,2}}){\mathbf{A}} - \frac{ {d_{{\rm{R}}}^2} } { {r_s^2}} {\mathbf{A}} {\rm{diag}} ({{\mathbf{v}}_{t,2}}) ].
	\end{equation}
\end{subequations}Based on the following derivations, a new expression of SPEB, which unveils the impact of antenna interval and antenna number, is given by \textbf{Proposition~\ref{Pro:FIM_expression}}.
    \begin{proposition}\label{Pro:FIM_expression}
    	By exploiting the symmetry and Hermitian matrix property, $ {\rm{SPEB}}$ is given by 
    	\begin{equation}
    		\begin{split}
    			& {\rm{SPEB}} \left(\{ {\mathbf{R}}_{x}, {\mathbf{q}}_{r}, {\bm{\Theta}}_{r} \} ; {\bm{\eta}} \right)  \\			= & \frac{ {p}_{x}^{2} {p}_{y}^{2}  }{ \left({p}_{x}+{p}_{y}\right)^{2} } \sum\limits_{i = 1}^2 \left\{  { [ {\mathbf{J}}_{ {{\tilde{{\bm{\theta}}}}_s} {\tilde{{\bm{\theta}}}}_{s} } ]}_{(i,i)}^{-1} { {[ {\tilde{{\mathbf{T}}}} ]}_{(i,i)} } \right\},
    		\end{split}
    	\end{equation}where $ {\tilde{{\mathbf{T}}}}=\frac{ {{{({p_x} + {p_y})}^2}} } { {p_x^2p_y^2} } {{{\mathbf{T}}^{-1}}{{({{\mathbf{T}}^H})}^{ - 1}}}$. The expressions of $\left[{\mathbf{J}}_{ {{\tilde{{\bm{\eta}}}}}  }\right]_{\left(i,i\right)}, i=1,2 $ are given in \eqref{J_expression}, where $ {\left\| {{{\mathbf{v}}_{r,2}}} \right\|^2_2} =\frac{{M_r ({M_r} + 1)({M_r} + 2)(3M_r^2 + 6{M_r} - 4) }}{{240}} $. $a\left( {\mathbf{R}}_{{{\bar {\mathbf{X}}}_r}} \right)={\Re \left[ {\rm{Tr}}( {\mathbf{R}}_{{{\bar {\mathbf{X}}}_r}} {\bar{\mathbf{A}}_t} ) \right] } $ denotes the correlation factor between the joint BF covariance matrix ${\mathbf{R}}_{{{\bar {\mathbf{X}}}_r}} $ and the outer product ${\bar{\mathbf{A}}_t} $ of ${\bm{\alpha}}_{t} $. $b\left( {\mathbf{R}}_{{{\bar{\mathbf{X}} }_r}}, {\mathbf{v}}_{t,i} \right)={\Re [ {\rm{Tr}}( {\mathbf{R}}_{{{\bar{\mathbf{X}} }_r}} {\rm{diag}} ({{\mathbf{v}}_{t,i}}) {\bar{\mathbf{A}} _t} {\rm{diag}}\left( {\mathbf{v}}_{t,i}  \right)) ]}$ denotes the correlation factor between the joint BF covariance matrix ${\mathbf{R}}_{{{\bar {\mathbf{X}}}_r}}  $ and the outer product of ${\bm{\alpha}}_{t} \odot {\mathbf{v}}_{t,i} $. $c\left( {\mathbf{R}}_{{{\bar{\mathbf{X}} }_r}} \right)={\Re [ {\rm{Tr}}({\rm{diag}}\left( {{{\mathbf{v}}_{t,2}}} \right) {\mathbf{R}}_{{{\bar{\mathbf{X}} }_r}} {\bar{\mathbf{A}}_t} ) + {\rm{Tr}}( {\mathbf{R}}_{{{\bar{\mathbf{X}} }_r}} {\rm{diag}}\left( {\mathbf{v}}_{t,2} \right) {\bar{\mathbf{A}}_t})]}$ denotes the cross-correlation between ${\mathbf{R}}_{{{\bar {\mathbf{X}}}_r}} {\bar{\mathbf{A}}_t} $ and ${\mathbf{v}}_{t,2} $.
    \end{proposition}
    \begin{proof}
    	Please see the details in \textbf{Appendix~\ref{App:FIM_expression}}.
    \end{proof}
    
    \begin{remark}
    	Based on \textbf{Proposition~1}, the FIM of ${\bm{\eta}}$ in NF propagation with the symmetry STAR array and sensors array can be approximated as a diagonal matrix. The SPEB of ${\bm{\eta}}$ depends on the joint BF design $\left\{{\mathbf{R}}_{x}, {\bm{\Theta}}_{r} \right\}$, the STAR/sensor array design, the location of estimated target and carrier frequency. The SPEB reduces with the increase of ${M}_{r}$ and ${d}_{R}, {d}_{s}$ while the SPEB degrades with the increase of ${r}_{s}$.The number of sensor elements has a more obvious impact than the number of STAR elements. Meanwhile, the location of ST also impacts the radio sensing performance. The SPEB monotonically decreases with the increase of $r_{s}$ due to the less power of the received echo signal. The angle estimation performance is inversely proportional to the cosine of the angle.
    \end{remark}
    
    \begin{figure*}
    	\begin{subequations}\label{J_expression}
    		\begin{equation}
    			{ [ {\mathbf{J}}_{ {{\tilde{{\bm{\eta}}}}}  }   ]}_{(1,1)}=\frac{ {2{{\left| {{\alpha_s}} \right|}^2}} {\pi^2}{{\sin}^2}{\theta_z} L {M_r} }{ {3 {\sigma^2} {\lambda}_c^2 } } \left[ { (M_r+1)  (M_r+2)} {d}_s^{2} a\left( {\mathbf{R}}_{{{\bar {\mathbf{X}}}_r}} \right)+ 12 d_{{\rm{R}}}^{2} b\left( {\mathbf{R}}_{{{\bar{\mathbf{X}} }_r}}, {\mathbf{v}}_{t,1} \right) \right],
    		\end{equation}
    	  \begin{equation}
    	  	[ {\mathbf{J}}_{ {{\tilde{{\bm{\eta}}}}}  } ]_{(2,2)}= - \frac{ {d}_s^{2} {\pi^2} d_{\rm{R}}^2 {{\left| {\alpha_s} \right|}^2} L{M_r}\left(M_r+1\right) (M_r+2) c\left( {\mathbf{R}}_{{{\bar{\mathbf{X}} }_r}} \right)} { 6 \sigma^{2} {\lambda_c^2} {\tilde{r}}_s^2 {r}_s^2 } + \frac{ {2{{\left| {\alpha_s} \right|}^2} {\pi^2} L{M_r}} } { \sigma^{2} {\lambda_c^2} } \left\{ \frac{ { {d}_s^4  \left\| {{{\mathbf{v}}_{r,2}}} \right\|_2^2} a\left( {\mathbf{R}}_{{{\bar {\mathbf{X}}}_r}} \right) } {{\tilde{r}_s^4}}  +\frac{ {d_{\rm{R}}^4}b\left( {\mathbf{R}}_{{{\bar{\mathbf{X}} }_r}}, {\mathbf{v}}_{t,2} \right)} {{r}_s^4}  \right\}. 
    	  \end{equation}
    	\end{subequations}
       \hrulefill
    \end{figure*}
\vspace{-0.6cm}    
\subsection{Problem Formulation}
To alleviate the signal attenuation of multi-hop links, the radio sensors are installed at the STARS to conduct the radio sensing functionality. With the high frequency in NF propagation and signal processing requirement of sensors at the STARS, the dense sensor deployment inevitably induces high power consumption and hardware cost~\cite{QingqingWu_21TCOM,Wanting_24TCOM}. Considering the limited aperture size and the high cost of sensor deployment, we investigate how to balance the trade-off between the SPEB and the sensor deployment cost. To unify the scale of the SPEB and sensor deployment cost, a weighted cost function minimization problem is formulated, which is composed of the normalized SPEB and the normalized sensor deployment cost. Let ${M}_{0} $ denotes the budget of the sensor element, the corresponding joint sensor deployment and BF design problem can be formulated as follows
  \begin{subequations}\label{Eq:PF_ori}
  	\begin{align}
  		&\label{Weighted objective function} \underset{ \left\{ \begin{smallmatrix} 
  				{{\mathbf{w}}_{s}}, {{\mathbf{v}}_{c}}, \left\{ {\mathbf{{q}}_{l}}\right\}, \\
  				d_{s}, M_{r}
  			\end{smallmatrix} \right\} } {\mathop{\min }} {\omega}_{0}\frac{{\rm{SPEB}}}{ \epsilon_{0} } + (1-{\omega}_{0}) \frac{{M}_{r}}{ {M}_{0} }  \,\  \\    
  		\label{Communication QoS_NOMA_CU}& {\rm{s.t.}}\;\;   {R}_{1} \ge {\bar{R}}_{min}, \\
  		\label{Power max}& \quad \quad  {\mathbf{w}}_{s}^{H} {\mathbf{w}}_{s} + {\mathbf{v}}_{c}^{H} {\mathbf{v}}_{c}  \le  P_{max}, \\
  		\label{Auxiliary Variables}& \quad \quad {\mathbf{R}}_{x}={\mathbf{w}}_{s} {\mathbf{w}}_{s}^{H}+ {\mathbf{v}}_{c} {\mathbf{v}}_{c}^{H}, \; {\bm{\Theta}}_{l}={\rm{diag}}\left({\mathbf{q}}_{l}  \right), l \in \left\{r,t\right\}, \\
  		\label{deployment} & \quad \quad {M}_{r} {d}_{s} \leq {M} d, \quad {M}_{r} \ge 2, {M}_{r} \in {\mathbb{N}}_{+}, \\ 
  		\label{phrase_shift1}& \quad \quad {\beta}_{r,m}^{2}+{\beta}_{t,m}^{2} \le 1,  \quad \quad \forall m \in {\mathcal{M}}, \\
  		\label{theta_phrase}& \quad \quad \theta_{l,m} \in \left[ {0,2\pi} \right), \;  l \in \left\{ r,t\right\},m \in {\mathcal{M}}, \\
  		\label{theta_amplitude2}& \quad \quad 0 < {\beta}_{t,m}, {\beta}_{r,m} < 1, m \in {\mathcal{M}},
  	\end{align}
  \end{subequations}where ${\epsilon}_{0}$ denotes the predefined sensing precision and ${\omega}_{0} \in \left(0,1\right)$ denotes the weight for sensing performance. The second term in the cost function (CF) \eqref{Weighted objective function} denotes the deployment cost for the sensor element. ${\bar{R}}_{min}$ denotes the required achievable data rate of CU, $P_{max}$ denotes the total transmission power at the BS. \eqref{phrase_shift1}-\eqref{theta_amplitude2} depict the constraints of STARS responses. Due to the quadratic form, coupled variables and integer variables in \eqref{Weighted objective function}-\eqref{Power max}, \eqref{Eq:PF_ori} is a non-convex mixed integer problem, which is difficult to obtain an optimal solution. 

    \vspace{-0.2cm}
    \section{Proposed Algorithm}
    To figure out the solution of the formulated problem~\eqref{Eq:PF_ori}, we split the optimization variables into three blocks: 1) The sensor deployment block: $\left\{ d_{s}, {M}_{r}  \right\}$; 2) The active beamforming block: $\left\{ {\mathbf{w}}_{s}, {\mathbf{v}}_{c} \right\} $ ; 3) The passive beamforming block: $\left\{ {\mathbf{q}}_{r}, {\mathbf{q}}_{t} \right\} $. 
\vspace{-0.4cm}
\subsection{The Sensor Deployment Design} 
With the other blocks fixed, the original  problem~\eqref{Eq:PF_ori} w.r.t. $\left\{ d_{s}, {M}_{r} \right\} $ can be expressed as
\begin{subequations}\label{Eq_Sensor_Deployment:PF_ori}
	\begin{align}
		&\label{Eq_Sensor_Deployment:Weighted objective function} \underset{ \left\{ \begin{smallmatrix} 
				d_{s}, M_{r}
			\end{smallmatrix} \right\} } {\mathop{\min }} {\omega}_{0}\frac{ {\rm{SPEB}} } { \epsilon_{0} } + (1-{\omega}_{0}) \frac{{M}_{r} }{{M}_{0}}   \,\  \\    
		\label{Eq_Sensor_Deployment:deployment} & {\rm{s.t.}}\;\;  {M}_{r} {d}_{s} \leq {M} d_{\rm{R}}, \\
		& \quad \quad   d_{s}>0,  \quad {M}_{r} \ge 2,  {M}_{r} \in {\mathbb{N}}_{+}.
	\end{align}
\end{subequations}Constrained by the coupled variables in~\eqref{Eq_Sensor_Deployment:deployment}, \eqref{Eq_Sensor_Deployment:PF_ori} is a mixed integer non-linear programming problem, which is difficult to obtain the optimal solution during the polynomial time. To deal with this dilemma, we introduce the auxiliary variables ${\tilde{d}}_{s}={d}_{s}^{2} $. The corresponding subproblem w.r.t. ${\tilde{d}}_{s}$ with the fixed ${M}_{r}$ is given by: 
\begin{subequations}\label{Eq_Sensor_Deployment:Sub_interval}
	\begin{align}
		&\label{Eq_Sensor_Deployment_Sub_interval:Weighted objective function} \underset{ \left\{  {\tilde{d}}_{s} \right\} } {\mathop{\min }} \frac{ {\omega}_{0}  }{ \epsilon_{0} } {\tilde{ {\rm{SPEB}} }} \left( {\tilde{d}}_{s}\right) \,\  \\    
		\label{Eq_Sensor_Deployment_Sub_interval:deployment} & {\rm{s.t.}}\;\; \frac{ {M}^{2} d_{{\rm{R}}}^{2} }{ 4 } \leq {\tilde{{d}}}_{s} \leq \frac{ {M}^{2} d_{{\rm{R}}}^{2} }{ {M}_{r}^{2} }.
	\end{align}
\end{subequations}
Let ${{\rm{C}}}_{11} \buildrel \Delta \over = \frac{ {2{{\left| {{\alpha_s}} \right|}^2}} {  {\pi^2}{{\sin}^2}{\theta _z}{M_r}({M_r} + 1)({M_r} + 2)}L a\left( {\mathbf{R}}_{{{\bar {\mathbf{X}}}_r}} \right) } { {3 {\sigma^2} {\lambda}_c^2 } }, {{\rm{C}}}_{10}\buildrel \Delta \over =\frac{ {8{{\left| {{\alpha _s}} \right|}^2}{\pi^2}{d}_{\rm{R}}^2 {{\sin }^2}{\theta_z}L{M_r}} b\left( {\mathbf{R}}_{{{\bar{\mathbf{X}} }_r}}, {\mathbf{v}}_{t,1} \right) }{ {{\sigma ^2}\lambda _c^2} }$, ${{\rm{C}}_0} \buildrel \Delta \over = \frac{ {2d_{\rm{R}}^4{{\left| {{\alpha _s}} \right|}^2}{\pi ^2}L{M_r}b \left( {{{\mathbf{R}}_{{{\bar{\mathbf{X}} }_r}}}}, {\mathbf{v}}_{t,2} \right)} } { {r_s^4{\sigma^2}\lambda _c^2} }$, ${{\rm{C}}_1} \buildrel \Delta \over = \frac{ {{\pi ^2}d_{\rm{R}}^2{{\left| {{\alpha _s}} \right|}^2}L{M_r} {(M_r^2 + 3{M_r} + 2)}c\left( {{{\mathbf{R}}_{{{\bar{\mathbf{X}} }_r}}}} \right)} } { {6{\sigma^2} \lambda_c^2 {\tilde{r}}_s^2 r_s^2} }$, and ${{\rm{C}}_2} \buildrel \Delta \over = \frac{ {{\left| {{\alpha_s}} \right|}^2}{\pi ^2}{M_r} {(M_r^2 + 3{M_r} + 2)} \left( 3{M}_{r}^{2}+6{M}_{r}-4 \right) L a\left( {{{\mathbf{R}}_{{{\bar{\mathbf{X}} }_r}}}} \right) } { 120 {\tilde{r}}_z^4{\sigma ^2}\lambda _c^2 }$, the diagonal entries of ${{\mathbf{{\mathbf{J}} }}}_{ {{\tilde{{\bm{\eta}}}}}  }$ are expressed as
\begin{subequations}
	\begin{equation}
		{ [ {\mathbf{J}}_{ {{\tilde{{\bm{\eta}}}}}  }]}_{(1,1)}= {\rm{C}}_{11} {\tilde{d}}_{s}+{\rm{C}}_{10}, 
	\end{equation}
	\begin{equation}
		\begin{split}
			[ {\mathbf{J}}_{ {{\tilde{{\bm{\eta}}}}}  }]_{(2,2)}& = {\rm{C}}_{2} {\tilde{d}}_{s}^{2}-{\rm{C}}_{1} {\tilde{d}}_{s}+{\rm{C}}_{0}  \\
			&= {\rm{C}}_{2} \left( {\tilde{d}}_{s}-\frac{{\rm{C}}_{1}} {2{\rm{C}}_{2}} \right)^{2}+{\rm{C}}_{0}-\frac{{\rm{C}}_{1}^{2}} {4{\rm{C}}_{2}}.
		\end{split}
	\end{equation}
\end{subequations}
The problem \eqref{Eq_Sensor_Deployment:Sub_interval} can be equivalently simplified as follows
\begin{subequations}\label{Eq_Sensor_Deployment:Sub_interval_Noauxiliary}
	\begin{align}
		&\label{Eq_Sensor_Deployment_Sub_interval:Weighted objective function_Sim} \underset{ \left\{  {\tilde{d}}_{s} \right\} } {\mathop{\min }} \frac{ {[\tilde{{\mathbf{T}}}]}_{(1,1)} } { {\rm{C}}_{11} {\tilde{d}}_{s}+ {\rm{C}}_{10} } + \frac{ {{{[\tilde{{\mathbf{T}}} ]}_{(2,2)}}} } { {\rm{C}}_{2} {\tilde{d}}_{s}^{2}-{\rm{C}}_{1} {\tilde{d}}_{s}+{\rm{C}}_{0} } \,\  \\    
		\label{Eq_Sensor_Deployment_Sub_interval_Sim:deployment} & {\rm{s.t.}}\;\; \frac{ {M}^{2} d_{{\rm{R}}}^{2} }{ {M}_{r}^{2} } \leq {\tilde{{d}}}_{s} \leq \frac{ {M}^{2} d_{{\rm{R}}}^{2} }{ 4 } .
	\end{align}
\end{subequations}Due to the non-convexity of the second term in the objective function, \eqref{Eq_Sensor_Deployment_Sub_interval:Weighted objective function_Sim} is a non-convex problem. However, it is noted that the first term in the objective function is monotonically decreasing w.r.t. ${\tilde{d}}_{s}$. Thus, we have the following proposition: 
\begin{proposition} \label{App_Opt_ds}
	If $ \frac{ (3M_r^2 + 6{M_r} - 4)  a\left( {{{\mathbf{R}}_{{{\bar{\mathbf{X}}}_r}}}} \right)  b\left( {\mathbf{R}}_{{{\bar{\mathbf{X}}}_r}}, {\mathbf{v}}_{t,2}  \right) } {5} \ge \frac{ {\left( {M_r^2+3{M}_{r}+2} \right)} {c^2} \left( {{{\mathbf{R}}_{{{\bar{\mathbf{X}} }_r}}}} \right) } {12}$ and $ \left| {\frac{{{{\rm{C}}_1}}}{{2{{\rm{C}}_2}}} - \frac{ {M}^{2} d_{{\rm{R}}}^{2} }{4}  } \right| < \left| { \frac{{{{\rm{C}}_1}}}{{2{{\rm{C}}_2}}} - \frac{ {M}^{2} d_{{\rm{R}}}^{2} }{ {M}_{r}^{2} } } \right| $, the optimal solution ${d}_{s}^{*}$ of \eqref{Eq_Sensor_Deployment_Sub_interval:Weighted objective function} can be expressed as
	\begin{equation}\label{opt_ds}
		{\tilde{d}}_{s}^{*}=\left(\frac{ {M} d_{{\rm{R}}}} { {M}_{r} } \right)^{2}. 
	\end{equation}
\end{proposition}
\begin{proof}
	Please see the details in {\textbf{Appendix~\ref{App:Opt_ds}}}. 
\end{proof}
\begin{remark}
	If the predefined BF design and the number of sensor elements satisfied the perquisite condition in \textbf{Proposition~\ref{App_Opt_ds}}, the best radio sensing performance can be achieved if the aperture size of the STAR elements is equal to the aperture size of the sensor elements. 
\end{remark}

However, if the conditions in \textbf{Proposition~\ref{App_Opt_ds}} are not satisfied, the optimal solution of~\eqref{Eq_Sensor_Deployment_Sub_interval:Weighted objective function_Sim} is difficult to obtained. Due to the quadratic form in the denominator of the second term in the objective function, problem \eqref{Eq_Sensor_Deployment:Sub_interval_Noauxiliary} is a non-convex problem. To figure out the optimal solution in the polynomial time, we introduce the auxiliary variables ${\bar{d}}_{s}$, which satisfies
\begin{equation}\label{InterSub_bar_exp}
	{\bar{d}}_{s} \buildrel \Delta \over  = {\tilde{d}}_{s}^{2} - \frac{{{{\rm{C}}_1}}} {{{{\rm{C}}_2}}} {\tilde{d}_s} + \frac{{{{\rm{C}}_0}}} {{4{{\rm{C}}_2}}}={\left( {{\tilde{d}}_s} - \frac{{{{\rm{C}}_1}}} {{2{{\rm{C}}_2}}} \right)^2} + \frac{{ {{\rm{C}}_0}{{\rm{C}}}_{2} - {{\rm{C}}}_{1}^{2} }}{ 4{{\rm{C}}}_{2}^{2} }. 
\end{equation}The problem \eqref{Eq_Sensor_Deployment_Sub_interval:Weighted objective function_Sim} can be equivalently converted into
\begin{subequations}\label{Eq_Sensor_Deployment:Weighted objective function_Sim}
	\begin{align}
		& \underset{ \left\{  {\tilde{d}}_{s}, {\bar{d}}_{s} \right\} } {\mathop{\min }} \frac{ {[\tilde{{\mathbf{T}}}]}_{(1,1)} } { {{\rm{C}}_{11}} {{\tilde{d}}_s} +{\rm{C}}_{10} } + \frac{ {{{[\tilde{{\mathbf{T}}} ]}_{(2,2)}}} } { {{\rm{C}}_2} {\bar{d}}_{s} } \,\  \\    
		\label{Eq_Sensor_Deployment:auxiliary} &  {\rm{s.t.}}\;\; 
		\frac{{ {{\rm{C}}_0}{{\rm{C}}}_{2} - {{\rm{C}}}_{1}^{2} }}{ 4{{\rm{C}}}_{2}^{2} } \le {\bar{d}}_{s} \le {\tilde{d}}_{s}^{2}-\frac{ {{\rm{C}}}_{1} }{{{\rm{C}}}_{2}}{\tilde{d}}_{s}+ \frac{ {{\rm{C}}}_{0} }{ 4{{\rm{C}}}_{2} },  \\
		& \quad \quad \eqref{Eq_Sensor_Deployment_Sub_interval_Sim:deployment}. 
	\end{align}
\end{subequations}
The equivalence can be explained in the following. Under the assumption that the optimal solution of \eqref{Eq_Sensor_Deployment:Weighted objective function_Sim} strictly satisfied the inequality constraints, the increase of ${\bar{d}}_{s}$ in \eqref{Eq_Sensor_Deployment:auxiliary} leads to the decrease of the objective function. This phenomenon contracts with the prior assumption. Therefore, the constraints of \eqref{Eq_Sensor_Deployment:Weighted objective function_Sim} are satisfied with strict equality for the corresponding optimal solution. To deal with the non-convexity of ${\bar{d}}_{s} \le {\tilde{d}}_{s}^{2}-\frac{ {{\rm{C}}}_{1} }{{{\rm{C}}}_{2}}{\tilde{d}}_{s}+ \frac{ {{\rm{C}}}_{0} }{ 4{{\rm{C}}}_{2} }$ in \eqref{Eq_Sensor_Deployment:auxiliary}, the first-order Taylor expansion is employed to derive the lower bound of ${\tilde{d}}_{s}^{2} $, which is given by
\begin{subequations}\label{Interval_SCA_lowerbound}
	\begin{equation}
		{\tilde{d}}_{s}^{2} \ge {\tilde{d}}_{s}^{\left( n \right)2} +2 {\tilde{d}}_{s}^{ \left(n\right) }\left( {\tilde{d}}_{s} -{\tilde{d}}_{s}^{ \left(n\right) } \right) \buildrel \Delta \over = \Gamma ( {\rm{ {\tilde{d}}_{s} }} ), 
	\end{equation}
	\begin{equation}
		{\bar{d}}_{s} \le \Gamma ( {\rm{ {\tilde{d}}_{s} }} ) - \frac{ {{\rm{C}}}_{1} }{{{\rm{C}}}_{2}}{\tilde{d}}_{s}+ \frac{ {{\rm{C}}}_{0} }{ 4{{\rm{C}}}_{2} }, 
	\end{equation}
\end{subequations}where ${\tilde{d}}_{s}^{n}$ denotes the given feasible point at the $n$th iteration of SCA algorithm. Therefore, \eqref{Eq_Sensor_Deployment:Weighted objective function_Sim} is converted into 
\begin{subequations}\label{Eq_Sensor_Deployment:Weighted objective function_cvx}
	\begin{align}
		& \label{Eq_Sensor_Deployment:aux_convex}\underset{ \left\{  {\tilde{d}}_{s}, {\bar{d}}_{s} \right\} } {\mathop{\min }} \frac{ {[\tilde{{\mathbf{T}}}]}_{(1,1)} } { {{\rm{C}}_{11}}{{\tilde{d}}_s}+{\rm{C}}_{10} } + \frac{ {{{[\tilde{{\mathbf{T}}} ]}_{(2,2)}}} } { {{\rm{C}}_2} {\bar{d}}_{s} } \,\  \\    
		& {\rm{s.t.}}\;\; 
		\frac{{ {{\rm{C}}_0} {\rm{C}}_{2} - {\rm{C}}_{1}^{2} }}{ 4{\rm{C}}_{2}^{2} } \le {\bar{d}}_{s} \le \Gamma ({\rm{ {\tilde{d}}_{s} }}) - \frac{ {\rm{C}}_{1} }{{\rm{C}}_{2}}{\tilde{d}}_{s}+ \frac{ {\rm{C}}_{0} }{ 4{\rm{C}}_{2} },  \\
		& \quad \quad \eqref{Eq_Sensor_Deployment_Sub_interval_Sim:deployment}.
	\end{align}
\end{subequations}It is clear that \eqref{Eq_Sensor_Deployment:Weighted objective function_cvx} is a convex optimization problem, where the corresponding optimal solution can be found via CVX~\cite{cvx}. Besides, the problem \eqref{Eq_Sensor_Deployment:Weighted objective function_cvx} provides an upper bound for problem~\eqref{Eq_Sensor_Deployment:Weighted objective function_Sim} because of the global lower bound in \eqref{Interval_SCA_lowerbound}. The details of the SCA algorithm for the sensor interval subproblem can be found in line 4 -line 11 of \textbf{Algorithm~\ref{Alg_IntervalSub}}.
    \begin{algorithm}[tp] 
         	\caption{Joint sensor deployment algorithm for \eqref{Eq_Sensor_Deployment:PF_ori}}\label{Alg_IntervalSub}
     	\begin{algorithmic}[1]\label{alg_1}
     		\STATE Initialize feasible points $ \left\{  {\mathbf{{q}}^{\left( 0 \right)}_{l}}, {\mathbf{w}}_{s}^{ \left( 0\right) }, {\mathbf{v}}_{c}^{\left(0\right)} \right\}$ and the out layer iteration index $m=1$. 		
     		\STATE Initialize feasible points $ {M}_{r}^{ \left( 0 \right) } $.  
     		\REPEAT 
     		\IF {$ \frac{ (3M_r^2 + 6{M_r} - 4) a\left( {{{\mathbf{R}}_{{{\bar{\mathbf{X}}}_r}}}} \right)  b\left( {\mathbf{R}}_{{{\bar{\mathbf{X}}}_r}}, {\mathbf{v}}_{t,2}  \right) } {5} \ge \frac{ {\left( {M_r^2+3{M}_{r}+2} \right)} {c^2} \left( {{{\mathbf{R}}_{{{\bar{\mathbf{X}} }_r}}}} \right) } {12}$ \& $\left| {\frac{{{C_1}}}{{2{C_2}}} - \frac{ {M}^{2} d_{{\rm{R}}}^{2} }{4}  } \right| < \left| { \frac{{{C_1}}}{{2{C_2}}} - \frac{ {M}^{2} d_{{\rm{R}}}^{2} }{ {M}_{r}^{2} } } \right| $.   }
     		\STATE ${d}_{s}^{*}= \frac{ {M} d_{{\rm{R}}} }{ {M}_{r} } $
     		\ELSE 
     		\STATE The inner layer iteration index $n=1$. 
     		\REPEAT 
     		\STATE Calculate $ {\bar{d}}_{s}^{\left(n\right)} $ via \eqref{InterSub_bar_exp}. 
     		\STATE Solve the convex problem \eqref{Eq_Sensor_Deployment:Weighted objective function_cvx} to obtain the optimal solution ${\tilde{d}}_{s}^{*} $. 
     		\STATE  Update the feasible points by the current optimal solution:  ${\tilde{d}}_{s}^{\left(n+1\right)}={\tilde{d}}_{s}^{*}$.  
     		\UNTIL The relevant variation of objective value for \eqref{Eq_Sensor_Deployment:Weighted objective function_Sim} among consecutive iteration is below a predefined threshold ${\epsilon}$, ${d}_{s}^{*}=\sqrt{{\tilde{d}}_{s}^{\left(n\right)}} $.
     		\ENDIF
     		\STATE Calculate the optimal solution ${\tilde{M}}_{r}^{\left(m\right)}$ by solving \eqref{Eq_Sensor_Deployment_v1:Sub_Num_P2_GP}. 
     		\STATE Update the elements number ${M}_{r}^{\left(m\right)}= {\arg } \min \left[ {\rm{OF}} \left( \left\lfloor {{\tilde{M}}_r} \right\rfloor \right), {\rm{OF}} \left( \left\lceil {{\tilde{M}}_r} \right\rceil \right) \right]$. 
     		\UNTIL{The relevant variation of objective value for \eqref{Eq_Sensor_Deployment:PF_ori} among consecutive iteration is below a convergence threshold ${\epsilon}$, ${M}_{r}^{*}={M}_{r}^{\left(m\right)}$.  }
     		\STATE \textbf{Output: \; ${M}_{r}^{*} \; {d}_{s}^{*}$}. 
     	\end{algorithmic}
    \end{algorithm}
    
    For the optimization of ${M}_{r}$, the corresponding subproblem can be expressed as 
    \begin{subequations}\label{Eq_Sensor_Deployment:Sub_Num}
    	\begin{align}
    		&\label{Eq_Sensor_Deployment:Sub_Num_obj} \underset{ M_{r}
    		 }  { \mathop{\min } } \; {\rm{OF}} \left( {M}_{r}  \right) \;    \,\  \\    
    		\label{Eq_Sensor_Deployment:Sub_Num_con1} & {\rm{s.t.}}\;\; 2 \leq {M}_{r} \leq \frac{ {M} d_{\rm{R}} }{ {d}_{s} } , {M}_{r} \in {\mathbb{N}}_{+}, 
    	\end{align}
    \end{subequations}where 
    ${\rm{OF}} \left( {M}_{r} \right)= \frac{ {\omega}_{0} {p}_{x}^{2} {p}_{y}^{2} } {\epsilon_{0} \left({p}_{x}+{p}_{y}\right)^{2} } \sum\limits_{i = 1}^2 \left\{  { [ {\mathbf{J}}_{ {{\tilde{{\bm{\eta}}}}} }]}_{(i,i)}^{-1} { {[ {\tilde{{\mathbf{T}}}} ]}_{(i,i)} } \right\}+ (1-{\omega}_{0}) \frac{{M}_{r} }{{M}_{0}}.$ Due to the discrete variables and the signomial geometric property in the denominators, \eqref{Eq_Sensor_Deployment:Sub_Num} is a non-convex signomial geometric programming, which is NP-hard~\cite{cvx}. To deal with this question, we relax the integer variables as continuous variables ${\tilde{M}}_{r} >2, {\tilde{M}}_{r} \in {\mathbb{R}}$ and employ the geometric programming approximation to solve it. The subproblem can be derived as: 
    \begin{subequations}\label{Eq_Sensor_Deployment_v1:Sub_Num_P1}
    	\begin{align}
    		&\label{Eq_Sensor_Deployment_v1:} \underset{ {\tilde{M}}_{r}
    		}  { \mathop{\min} } \; \frac{{{\omega _0}p_x^2p_y^2}} {{{{\epsilon}_0}{{\left( {{p_x} + {p_y}} \right)}^2}}} f_{0}\left( {\tilde{M}}_{r} \right) + (1 - {\omega _0})\frac{{{M_r}}}{{{M_0}}}  \,\  \\    
    		\label{Eq_Sensor_Deployment_v1:Sub_Num_con1} & {\rm{s.t.}} \;\; 2 \leq {\tilde{M}}_{r} \leq \frac{ {M} d_{\rm{R}} }{ {d}_{s} } , {M}_{r} \in {\mathbb{R}}. 
    	\end{align}
    \end{subequations}$f_{0}\left( {\tilde{M}}_{r} \right) \buildrel \Delta \over = { \frac{ {3{\sigma ^2}\lambda _c^2{{[\tilde{\mathbf{T}}]}_{(1,1)}}} } { {2{{\left| {{{\bm{\alpha}}_s}} \right|}^2}{\pi^2}{{\sin }^2}{ {\theta}_z}L {f_{1}}({{\tilde M}_r})} } + 
    \frac{ {{\sigma^2}\lambda_c^2} {{{[\tilde{\mathbf{T}}]}_{(2,2)}}} }{ { {\left| {\alpha}_s \right|}^2}{\pi^2}L \left[{{f_{2, + }}({{\tilde M}_r}) - {f_{2, - }}({{\tilde M}_r})} \right]  } }$. 
    ${f_{1}}({{\tilde M}_r})=d_s^2a\left( {{{\mathbf{R}}_{{{\bar{\mathbf{X}} }_r}}}} \right) ( {\tilde{M}}_r^3 + 3 {\tilde{M}}_r^2 )+ \left[ {2d_s^2 a\left( {{{\mathbf{R}}_{{{\bar{\mathbf{X}}}_r}}}} \right) + 36d_{\rm{R}}^2b\left( {{{\mathbf{R}}_{{{\bar{\mathbf{X}} }_r}}},{{\mathbf{v}}_{t,1}}} \right)} \right]{{\tilde{M}}_r}$. Let ${{\rm{B}}}_{i}, i \in \left\{0,1,2\right\}$ denote the coefficients of ${\tilde{M}}_{r}$ with the expressions ${{\rm{B}}_0} = \frac{{2d_{\rm{R}}^4b\left( {{{\mathbf{R}}_{{{\bar{\mathbf{X}}}_r}}},{{\mathbf{v}}_{t,2}}} \right)}}{{r_s^4}}$, ${{\rm{B}}_1} =\frac{{d_s^2d_{\rm{R}}^2c\left( {{{\mathbf{R}}_{{{\bar{\mathbf{X}} }_r}}}} \right)}}{{6 {\tilde{r}}_s^2r_s^2}}$, and $	{{\rm{B}}_2} = \frac{{{\tilde{d}}_s^2 a\left( {{{\mathbf{R}}_{{{\bar{\mathbf{X}}}_r}}}} \right)}}{{120{\tilde{r}}_s^4}} $. If $c\left( {{{\mathbf{R}}_{{{\bar{\mathbf{X}}}_r}}}} \right) \geq 0$, $ f_{2,+} \left({\tilde{M}}_{r}\right)=3{{\rm{B}}_2} {\tilde{M}}_r^5 + 15{{\rm{B}}_2} {\tilde{M}}_r^4 + 20{{\rm{B}}_2} {\tilde{M}}_r^3 + {{\rm{B}}_0} {\tilde{M}}_r$ and $f_{2,-} \left({\tilde{M}}_{r}\right)={{\rm{B}}_1} {\tilde{M}}_r^3 + 3{{\rm{B}}_1} {\tilde{M}}_r^2 + (8{{\rm{B}}_2} + 2{{\rm{B}}_1}){\tilde{M}}_r $. When $c\left( {{{\mathbf{R}}_{{{\bar{\mathbf{X}}}_r}}}} \right) <0$, the ${f}_{2,+}\left( {\tilde{M}}_{r}\right)$ and $ {f}_{2,-} \left( {\tilde{M}}_{r}\right)$ are given by ${f}_{2,+}\left( {\tilde{M}}_{r}\right)=3{{\rm{B}}_2} {\tilde{M}}_r^5 + 15{{\rm{B}}_2} {\tilde{M}}_r^4 + ({{\rm{{\bar{B}}}}_1} + 20{{\rm{B}}_2}) {\tilde{M}}_r^3 + 3{{\rm{{\bar{B}}}}_1} {\tilde{M}}_r^2 + (2{{\rm{{\bar{B}}}}_1} + {{\rm{B}}_0}) {\tilde{M}}_r$ and ${f}_{2,-} \left( {\tilde{M}}_{r}\right)= 8{{\rm{B}}_2}{\tilde{M}}_r $. Recall that ${\tilde{M}}_{r} \in \left[2, \frac{ {M} d_{{\rm{R}}} }{ d_{s} } \right] $ and ${{\rm{B}}}_{i}>0, i \in \left\{1,2,3\right\}$, it's clear that $f_{2,+}\left( {\tilde{M}}_{r} \right)-f_{2,-}\left( {\tilde{M}}_{r} \right) > 0 $ and $f_{2,+}\left( {\tilde{M}}_{r} \right)>0, f_{2,-}\left(  {\tilde{M}}_{r} \right)>0$. Because of the term in $f_{1}\left({\tilde{M}}_{r}\right)$ and the signomial term of $f_{2,+}\left( {\tilde{M}}_{r} \right)-f_{2,-}\left({\tilde{M}}_{r}\right)$, the problem~\ref{Eq_Sensor_Deployment_v1:Sub_Num_P1} is not a standard GP. To deal with the signomial property in $f_{0}\left( {\tilde{M}}_{r} \right)$ and problem \eqref{Eq_Sensor_Deployment_v1:Sub_Num_P1}, we introduce the slack variables $x_{2}$ as: 
    \begin{equation}
    	f_{2,+}\left( {\tilde{M}}_{r}\right)- f_{2,-}\left( {\tilde{M}}_{r} \right) \ge x_{2} >0. 
    \end{equation}Thus, $\frac{ { {\sigma^2} {\lambda}_c^2} {{{[\tilde{\mathbf{T}}]}_{(2,2)}}}} { {{{\left| {{\alpha _s}} \right|}^2}{\pi ^2}L} \left[ f_{2,+}\left({\tilde{M}}_{r}\right) - f_{2,-} \left({\tilde{M}}_{r} \right) \right] }  \leq \frac{ { {\sigma^2} {\lambda}_c^2} {{{[\tilde{\mathbf{T}}]}_{(2,2)}}}} { {{{\left| {{\alpha _s}} \right|}^2}{\pi ^2}L}  {x}_{2} } $. The problem \eqref{Eq_Sensor_Deployment_v1:Sub_Num_P1} can be equivalently converted into 
    \begin{subequations}\label{Eq_Sensor_Deployment_v1:Sub_Num_P2}
    	\begin{align}
    		&\label{Eq_Sensor_Deployment_v1_P2:} \underset{ {\tilde{M}}_{r}, x_{2}
    		}  { \mathop{\min} } \; \frac{{{\omega _0}p_x^2p_y^2}} {{{{\epsilon}_0}{{\left( {{p_x} + {p_y}} \right)}^2}}} {f}_{0}^{'} \left( {\tilde{M}}_{r}, {x}_{2}\right) + (1 - {\omega _0})\frac{{{{\tilde{M}}_r}}}{{{M_0}}}  \,\  \\    
    		\label{Eq_Sensor_Deployment_v1_P2:Sub_Num_con1} & {\rm{s.t.}} \;\; \frac{ f_{2,-}\left( {\tilde{M}}_{r} \right)+x_{2} } { f_{2,+} \left( {\tilde{M}}_{r} \right) } \le 1, \;  \\
    		& \quad \quad \eqref{Eq_Sensor_Deployment_v1:Sub_Num_con1}, 
    	\end{align}
    \end{subequations}where ${f}_{0}^{'} \left( {\tilde{M}}_{r}, {x}_{2}\right)= \frac{ {3{\sigma^2} {\lambda}_c^2{{[\tilde{\mathbf{T}}]}_{(1,1)}}} } { {2{{\left| {{{\bm{\alpha}}_s}} \right|}^2}{\pi^2}{{\sin }^2}{ {\theta}_z}L} {f_1}({{\tilde M}_r}) } + \frac{ { {\sigma^2} {\lambda}_c^2} {{{[\tilde{\mathbf{T}}]}_{(2,2)}}}} { {{{\left| {{\alpha _s}} \right|}^2}{\pi ^2}L}  {x}_{2} }$. The equivalence is shown in \textbf{Proposition~\ref{App_GP_Equivalence}}. 
    \begin{proposition}\label{App_GP_Equivalence}
       If $ {\tilde{M}}_{r}^{*} $ is the optimal solution of~\eqref{Eq_Sensor_Deployment_v1:Sub_Num_P1}, $\left\{ {\tilde{M}}_{r}^{*},x_{2}^{*} \right\} $ is the optimal solution of~\eqref{Eq_Sensor_Deployment_v1:Sub_Num_P2}, if and only if $x_{2}^{*}$ satisfy:
       \begin{equation}\label{GP_slack_variables}
		   x_{2}^{*} = f_{2,+}\left( {\tilde{M}}_{r}^{*}\right)- f_{2,-}\left( {\tilde{M}}_{r}^{*} \right).
       \end{equation}
    \end{proposition}
    \begin{proof}
    	Please refer to \cite[Lemma 2]{SGP_20TWC}.
     \end{proof}
    To deal with the posynomial terms in \eqref{Eq_Sensor_Deployment_v1_P2:} and \eqref{Eq_Sensor_Deployment_v1_P2:Sub_Num_con1}, we adopt the arithmetic geometric mean to approximate $f_{1} \left( {\tilde{M}}_{r} \right), \; f_{2,+} \left( {\tilde{M}}_{r} \right) $ by its monomial lower bounds ${{\hat{f}}_{1}}({{\tilde{M}}_r})$ or ${{\hat{f}}_{2,+}}({{\tilde{M}}_r}) $. To simplify the presentation, we take $f_{1} \left({\tilde{M}}_{r}\right) $ as an example to illustrate the approximation. By defining ${u}_{1} = d_s^2 a\left( {{{\mathbf{R}}_{{{\bar{\mathbf{X}}}_r}}}} \right) {\tilde{M}}_r^2$, ${u}_{2}=3 d_s^2 a\left( {\mathbf{R}}_{{{\bar{\mathbf{X}}}_r}} \right) {{\tilde M}_r}$, $ {u}_{3}=\left[ {2d_s^2a\left( {{{\mathbf{R}}_{{{\bar{\mathbf{X}}}_r}}}} \right) + 36d_{\rm{R}}^2b\left( {{{\mathbf{R}}_{{{\bar{\mathbf{X}} }_r}}},{{\mathbf{v}}_{t,1}}} \right)} \right]$, ${f_{1}}({{\tilde{M}}_r}) $ can be expressed as 
    \begin{equation}
    		{f_{1}}({{\tilde{M}}_r})={\tilde{M}}_r \left[ {u}_{1}  + {u}_{2} + {u}_{3}\right]  \ge  {\tilde{M}}_r \prod\limits_{i = 1}^3 {{{\left( {\frac{{{u_i}}}{{{\lambda _i}}}} \right)}^{{\lambda _i}}}} \buildrel \Delta \over = {{\hat{f}}_{1}}({{\tilde{M}}_r}) 
    \end{equation}where ${\lambda}_{i}=\frac{ {u}_{i} \left( {\tilde{M}}_{r} \right) }{ \sum\limits_{i = 1}^3 {{u_i}\left( {\tilde{M}}_{r} \right)}  }, i \in \left\{1,2,3\right\}.$ Similarly, $f_{2,+}\left( {\tilde{M}}_{r} \right)$ is approximated by ${f_{2,+}}({{\tilde{M}}_r})={\tilde{M}}_r \left[ {\bar{{u}}}_{1} + {\bar{{u}}}_{2} + {\bar{{u} }}_{3}+ {\bar{u }}_{4} \right] \ge {\tilde{M}}_r \prod\limits_{i = 1}^4 {{{\left( {\frac{ {{\bar{u}}_i} }{{{{\bar{\lambda}}_i}}}} \right)}^{{{\bar{\lambda }}_i}}}} \buildrel \Delta \over = {{\hat{f}}_{2,+}}({{\tilde{M}}_r})$, where ${{\bar{u}}_1}=3{{\rm{B}}_2} {\tilde{M}}_r^4$, ${{\bar{u}}_2}=15{{\rm{B}}_2}\tilde M_r^3$, ${{\bar{u}}_3}=20{{\rm{B}}_2}\tilde M_r^2 $, and ${{\bar{u}}_4}={{\rm{B}}_0}$. ${\bar{\lambda}}_{i}$ are given by ${\lambda}_{i}=\frac{ {\bar{u}}_{i} \left( {\tilde{M}}_{r} \right) }{ \sum\limits_{i = 1}^3 {{{\bar{u}}_i}\left( {\tilde{M}}_{r} \right)}  }, i \in \left\{1,2,3\right\} $. Therefore, $\frac{ f_{2,-}\left( {\tilde{M}}_{r} \right)+x_{2} } { f_{2,+} \left( {\tilde{M}}_{r} \right) } \le 1 $ can be approximated by 
    \begin{equation}\label{Approximate_f}
    	\frac{ f_{2,-}\left( {\tilde{M}}_{r} \right)+x_{2} } { f_{2,+} \left( {\tilde{M}}_{r} \right) } \le \frac{ f_{2,-}\left( {\tilde{M}}_{r} \right)+x_{2} } { {\hat{f}}_{2,+} \left( {\tilde{M}}_{r} \right) } \le 1. 
    \end{equation} The problem \eqref{Eq_Sensor_Deployment_v1:Sub_Num_P2} is derived as
    \begin{subequations}\label{Eq_Sensor_Deployment_v1:Sub_Num_P2_GP}
    	\begin{align}
    		&\label{Eq_Sensor_Deployment_v1_P2_GP:} \underset{ {\tilde{M}}_{r}, x_{2}
    		}  { \mathop{\min} } \; \frac{{{\omega _0}p_x^2p_y^2}} {{{{\epsilon}_0}{{\left( {{p_x} + {p_y}} \right)}^2}}} {\bar{f}}_{0} \left( {\tilde{M}}_{r}, {x}_{2}\right) + (1 - {\omega _0})\frac{{{{\tilde{M}}_r}}}{{{M_0}}}  \,\  \\    
    		\label{Eq_Sensor_Deployment_v1_P2_GP:Sub_Num_con1} & {\rm{s.t.}} \;\;  \eqref{Eq_Sensor_Deployment_v1:Sub_Num_con1}, \eqref{Approximate_f},
    	\end{align}
    \end{subequations}where ${\bar{f}}_{0} \left( {\tilde{M}}_{r}, {x}_{2}\right) = \frac{ {3{\sigma^2} {\lambda}_c^2{{[\tilde{\mathbf{T}}]}_{(1,1)}}} } { {2{{\left| {{{\bm{\alpha}}_s}} \right|}^2}{\pi^2}{{\sin }^2}{ {\theta}_z}L} {{\hat{f}}_1}({{\tilde M}_r}) } + \frac{ { {\sigma^2} {\lambda}_c^2} {{{[\tilde{\mathbf{T}}]}_{(2,2)}}}} { {{{\left| {{\alpha _s}} \right|}^2}{\pi ^2}L}  {x}_{2} } $.
    Problem~\eqref{Eq_Sensor_Deployment_v1:Sub_Num_P2_GP} is a standard geometric programming (GP), which can be solved efficiently. Based on the continuous optimal solution ${\tilde{M}}_{r}^{*} $ of \eqref{Eq_Sensor_Deployment_v1:Sub_Num_P2}, the original integer solution $ {M}_{r}$ can be decided via the comparison of $ {\rm{OF}} \left( \left\lfloor {{\tilde{M}}_r} \right\rfloor \right)$ and ${\rm{OF}} \left( \left\lceil {{\tilde{M}}_r} \right\rceil \right)$. The details can be seen in line 13 - line 14 in \textbf{Algorithm~\ref{alg_1}}. 
    
    \subsection{The Beamforming Design}
    \subsubsection{The Active Beamforming Design}
    With the fixed $\left\{{M}_{r}, {d}_{s}, {\mathbf{q}}_{r}, {\mathbf{q}}_{t} \right\}$, the subproblem about the active BF can be equivalently simplified as
    \begin{subequations}\label{ActBF_Pro_Ori}
    	\begin{align}
    		& \label{ActBF_Pro_Ori:Obj}\mathop {\min } \limits_{ {{{\mathbf{w}}_s},{{\mathbf{v}}_c}} }  \frac{{{\omega_0}p_x^2p_y^2}} {{{{\epsilon}_0}{{\left( {{p_x} + {p_y}} \right)}^2}}} \sum\limits_{i = 1}^2 {\left\{ { [{{\mathbf{J}}_{ {{\tilde{{\bm{\eta}}}}}  }}]_{(i,i)}^{ - 1}  {[{\tilde{\mathbf{T}}}]}_{(i,i)} } \right\}} \\
    		& \;  {\rm{s.t.}}  \quad \eqref{Communication QoS_NOMA_CU}-\eqref{Auxiliary Variables}. 
    	\end{align}
    \end{subequations}Recall that ${\bm{\alpha}}_{t}^{H} {\bm{\Theta}}_{l}={\mathbf{q}}_{l}^{H} {\rm{diag}}\left( {\bm{\alpha}}_{t}^{H} \right)$, ${\rm{Tr}} \left( {\mathbf{R}}_{ {\bar{{\mathbf{X}}}}_{r} } {\bar{{\mathbf{A}}}}_{t} \right)={\rm{Tr}} \left( {\bm{\alpha}}_{t}^{H} {\mathbf{\Theta}}_{r} {\mathbf{H}}_{\rm{BR}} {\mathbf{R}}_{x} {\mathbf{H}}_{\rm{BR}}^{H} {\mathbf{\Theta}}_{r}^{H}     {\bm{\alpha}}_{t} \right)={\rm{Tr}} \left( {\rm{diag}}({\bm{\alpha}}_{t}^{H}) {\mathbf{H}}_{\rm{BR}} {\mathbf{R}}_{x} {\mathbf{H}}_{\rm{BR}}^{H} {\rm{diag}}({\bm{\alpha}}_{t}) {\mathbf{Q}}_{r} \right)$, where ${\mathbf{Q}}_{l}={\mathbf{q}}_{l} {\mathbf{q}}_{l}^{H}, l \in \left\{r,t\right\} $. The expressions of $ \left\{a\left( {\mathbf{R}}_{{{\bar {\mathbf{X}}}_r}} \right), b\left( {\mathbf{R}}_{{{\bar{\mathbf{X}} }_r}}, {\mathbf{v}}_{t,i} \right), c\left( {\mathbf{R}}_{{{\bar{\mathbf{X}} }_r}} \right)   \right\}$ can be derived as
    \begin{subequations}
    	\begin{equation}
    		a\left( {\mathbf{R}}_{{{\bar {\mathbf{X}}}_r}} \right)={\rm{Tr}} \left[ {{\mathbf{R}}_x} {{\bm{\Gamma}}_s} ( {\bar{{\mathbf{A}}}_t} ) \right] ,
    	\end{equation}
    	\begin{equation}
    		b\left( {\mathbf{R}}_{{{\bar{\mathbf{X}} }_r}}, {\mathbf{v}}_{t,i} \right)={\rm{Tr}} \left[ {{\mathbf{R}}_x} {{\bm{\Gamma}}_s} ( {{\bar{{\mathbf{A}}}}_t},{\mathbf{v}}_{t,i} ) \right],
    	\end{equation}
    	\begin{equation}
    		c\left( {\mathbf{R}}_{{{\bar{\mathbf{X}} }_r}} \right)={\rm{Tr}} \left[ {{\mathbf{R}}_x} {{{\bm{\Gamma}}^{'}}_s}( {\bar{{\mathbf{A}}}_t} ,{{\mathbf{v}}_{t,2}}) \right],
    	\end{equation}
    \end{subequations}where ${{\bm{\Gamma}}_s}({\bar{\mathbf{A}}_t}) = {\mathbf{H}}_{{\rm{BR}}}^H{\rm{diag}}( {{\bm{\alpha}}_t} ){{\mathbf{Q}}_r} {\rm{diag}}{\left( {{\bm{\alpha}}_t^H} \right)}{{\mathbf{H}}_{{\rm{BR}}}} \in {\mathbb{C}}^{N \times N}$, ${{{\bm{\Gamma}}^{'}}_s}( {\bar{{\mathbf{A}}}_t} , {{\mathbf{v}}_{t,2}})={\mathbf{H}}_{{\rm{BR}}}^H {\rm{diag}}({{\bm{\alpha}}_t}){{\mathbf{Q}}_r} {\rm{diag}} {\left( {{\mathbf{v}}_{t,2}} \odot {{\bm{\alpha}}_t}  \right)^H} {{\mathbf{H}}_{{\rm{BR}}}} \in {\mathbb{C}}^{N \times N} $. By denoting $ {{{\rm{\tilde{C}}}}_0} \buildrel \Delta \over = \frac{2{{\left| {{\alpha}_s} \right|}^2} {\pi^2} {{\sin }^2}{\theta_z}L{M_r}} {{3{\sigma ^2} {\lambda} _c^2}} $, ${{{\rm{\tilde C}}}_1} = \frac{{2d_s^4\left\| {{{\mathbf{v}}_{r,2}}} \right\|_2^2{{\left| {{\alpha _s}} \right|}^2}{\pi^2}L{M_r}}}{{\tilde{r}_s^4{\sigma^2}{\lambda}_c^2}} $, ${{{\rm{\tilde C}}}_2} = \frac{{2d_{\rm{R}}^4{{\left| {{\alpha _s}} \right|}^2}{\pi ^2}L{M_r}}}{{r_s^4{\sigma ^2}{\lambda}_c^2}} $, and ${{{\rm{\tilde C}}}_3} = \frac{{d_s^2 d_{\rm{R}}^2{{\left| {{\alpha_s}} \right|}^2}{\pi^2}L{M_r}\left( {{M_r} + 1} \right)({M_r} + 2)}} {{6 {\tilde{r}}_s^2 r_s^2{\sigma ^2} {\lambda}_c^2}}$, the objective function of \eqref{ActBF_Pro_Ori} can be expressed as $ [{{\mathbf{J}}_{ {{\tilde{{\bm{\eta}}}}}  }}]_{(1,1)} = {{\rm{\tilde{C}}}_0}  ({M_r} + 1)( {M_r} + 2) d_s^2 {\rm{Tr}} \left[ {{\mathbf{R}}_x} {{\bm{\Gamma}}_s} ( {\bar{{\mathbf{A}}}_t} ) \right] +{{\rm{\tilde{C}}}_0} 36 d_{\rm{R}}^2 {\rm{Tr}} \left[ {{\mathbf{R}}_x} {{\bm{\Gamma}}_s} ( {{\bar{{\mathbf{A}}}}_t},{\mathbf{v}}_{t,1} ) \right] $, and $	[{{\mathbf{J}}_{ {{\tilde{{\bm{\eta}}}}} }}]_{(2,2)}  = {{{\rm{\tilde{C}}}}_1} {\rm{Tr}} \left[ {{\mathbf{R}}_x}{{\bm{\Gamma}}_s}( {\bar{{\mathbf{A}}}_t} )  \right] + {{{\rm{\tilde{C}}}}_2} {\rm{Tr}}\left[ {{\mathbf{R}}_x} {{\bm{\Gamma}} _s} ( {\bar{{\mathbf{A}}}_t}, {{\mathbf{v}}_{t,2}}) \right] -{{{\rm{\tilde{C}}}}_3} {\rm{Tr}} \left[ {{\mathbf{R}}_x} {{{\bm{\Gamma}}^{'}}_s}( {\bar{{\mathbf{A}}}_t} ,{{\mathbf{v}}_{t,2}}) \right]  $. 
    To deal with the optimization variables in the denominators of \eqref{ActBF_Pro_Ori:Obj}, we introduce the auxiliary variable ${\mathbf{U}} \in {\mathbb{C}}^{2\times2}$ to simplify the problem into a more tractable form, which satisfies
    \begin{equation}\label{U_initial}
    	{\rm{diag}} \left( \left[{   { \frac{ [{{\mathbf{J}}_{ {{\tilde{{\bm{\eta}}}}}  }}]_{(1,1)} }{ {[{\tilde{\mathbf{T}}}]}_{(1,1)} }
    	}  }, { { \frac{ [{{\mathbf{J}}_{ {{\tilde{{\bm{\eta}}}}} }}]_{(2,2)} }{ {[{\tilde{\mathbf{T}}}]}_{(2,2)} }
    	}  } \right]  \right) \succeq {\mathbf{U}}.
    \end{equation}For the quadric form of $\left\{{\mathbf{w}}_{s}, {\mathbf{v}}_{c} \right\} $ in both the objective and constraints function, we introduce the auxiliary variables $ {\mathbf{V}}_{c}={\mathbf{v}}_{c}{\mathbf{v}}_{c}^{H} \in {\mathbb{C}}^{N \times N}, {{\mathbf{R}}_{s0}} = {{\mathbf{w}}_s}{\mathbf{w}}_s^H \in {\mathbb{C}}^{N \times N} $. The subproblem~\eqref{ActBF_Pro_Ori} can be equivalently converted into
    \begin{subequations}\label{ActBF_Pro_Ori_Maxtrix}
    	\begin{align}
    		& \mathop {\min} \limits_{ \left\{ {{{\mathbf{V}}_c}}, {{\mathbf{R}}_{s0}}, {\mathbf{U}} \right\}} {\rm{Tr}}\left( { {{ { {\mathbf{U}}}^{- 1}}} } \right) \\
    		\label{ActBF_Pro_Ori_Maxtrix_C1}&  {\rm{s.t.}} \;\; {\rm{Tr}} \left( {{{\mathbf{V}}_c} {{\bm{\Gamma}}_c}} \right) - ({2^{{{\bar{R}}_{min}}}} - 1) [{\rm{Tr}}\left( {{{\mathbf{R}}_{s0}} {{\bm{\Gamma}}_c}} \right) + {\sigma^2}] \ge 0, \\
    		\label{ActBF_Pro_Ori_Maxtrix_C2}& \quad \quad {\rm{Tr}}\left( {\mathbf{R}}_{x}\right) \le {P_{max}},\\
    		\label{ActBF_Pro_Ori_Maxtrix_C3}& \quad  \quad {{\mathbf{R}}_x} - {{\mathbf{V}}_c} \succeq {\mathbf{0}}, \quad {\mathbf{V}}_{c}  \succeq {\mathbf{0}},\\
    		& \quad \quad {\rm{rank}} \left( {\mathbf{V}}_{c} \right)=1, \\
    		& \quad \quad \eqref{U_initial}, 
    	\end{align}
    \end{subequations}where ${{\bm{\Gamma}}_c} = {\mathbf{H}}_{{\rm{BR}}}^H 
    {\rm{diag}}({\mathbf{g}}_c){{\mathbf{Q}}_t} {\rm{diag}}({\mathbf{g}}_c^H)
    {{\mathbf{H}}_{{\rm{BR}}}} \in {\mathbb{C}}^{N \times N} $. It is worth noting that the problem \eqref{ActBF_Pro_Ori_Maxtrix} is an equivalent transformation of \eqref{ActBF_Pro_Ori}. To prove this equivalence, we first assume that $ {\mathbf{U}}$ holds strict inequality for \eqref{U_initial}. With the increased ${\mathbf{U}}$, the constraint still holds while the objective function increases until \eqref{U_initial} reaches the equality. It is clear that the constraints \eqref{ActBF_Pro_Ori_Maxtrix_C1}-\eqref{ActBF_Pro_Ori_Maxtrix_C3} in problem~\eqref{ActBF_Pro_Ori_Maxtrix} are linear matrix inequalities of $\left\{{\mathbf{R}}_{s0}, {\mathbf{V}}_{c} \right\} $ and the objective function is mono-decreasing with ${\mathbf{U}}$~\cite{boyd_CVX}. To deal with the non-convex rank one constraints, we employ the SDR approach by omitting the rank constraints. The problem~\eqref{ActBF_Pro_Ori_Maxtrix} can be converted into the convex semidefinite programming (SDP) problem~\eqref{ActBF_Pro_Ori_Maxtrix_SDR}, which can be solved via CVX~\cite{cvx}.
    \begin{subequations}\label{ActBF_Pro_Ori_Maxtrix_SDR}
    	\begin{align}
    		& \mathop {\min} \limits_{ \left\{ {{{\mathbf{R}}_{s0}},{{\mathbf{V}}_c}},{\mathbf{U}} \right\}} {\rm{Tr}}\left( { {{ { {\mathbf{U}}}^{- 1}}} } \right) \\
    		& {\rm{s.t.}} \; \; \eqref{ActBF_Pro_Ori_Maxtrix_C1}-\eqref{ActBF_Pro_Ori_Maxtrix_C3}.
    	\end{align}
    \end{subequations}The rank one optimal solution can be achieved as shown in the {\textbf{Proposition~\ref{Rank_one_solution} }}. 
    \begin{proposition}\label{Rank_one_solution}
    	With the global optimal solution $\left\{ {\mathbf{{\bar{R}}}}_{s0}, {\mathbf{{\bar{V}}}_{c}} \right\} $ of problem~\eqref{ActBF_Pro_Ori_Maxtrix_SDR}, there is a optimal solution $\left\{ {\mathbf{R}}_{s0}^{*},  {\mathbf{V}}_{c}^{*} \right\} $) of problem~\eqref{ActBF_Pro_Ori_Maxtrix}, which can be expressed as
    	\begin{equation}\label{Rank_One_Solution}
    		{\mathbf{R}}_{s0}^{*}={\mathbf{{\bar{R}}}}_{s0}, \quad  {\mathbf{v}}_{c}^{*}=\left( {\mathbf{u}}_{c}^{H} {\mathbf{{\bar{V}}}}_{c} {\mathbf{u}}_{c} \right)^{-1/2} {\mathbf{{\bar{V}}}}_{c} {\mathbf{u}}_{c},
    	\end{equation}where ${\mathbf{u}}_{c}={\mathbf{H}}_{{\rm{BR}}}^{H} {\mathbf{\Theta}}_{t}^{H} {\mathbf{g}}_{c} \in {\mathbb{C}}^{N \times 1}$ denotes the equivalent channel vector of the CU.
    \end{proposition}
    \begin{proof}
    	Please refer to~\cite[Theorem 1]{XiangLiu_20TSP}.
    \end{proof}
    
    \subsubsection{The Passive Beamforming} 
    With the fixed active BF $\left\{{\mathbf{V}}_{c}, {\mathbf{R}}_{s0} \right\}$ and the auxiliary variable ${\mathbf{U}}$, the subproblem w.r.t. passive BF can be equivalently expressed as
    \begin{subequations}\label{PasBF_Pro_Ori}
    	\begin{align}
    		& \quad \quad  {\rm{Find}} \quad \left\{ {{{\mathbf{Q}}_r}}, {{\mathbf{Q}}_{t}} \right\}   \\
    		\label{PasBF_Pro_Ori_SINR} & {\rm{s.t.}} \; \; {\rm{Tr}}\left( {{{\bm{\Upsilon}}_c}({{\mathbf{g}}_c}){{\mathbf{Q}}_t}} \right) \ge ({2^{{{\bar{R}}_{min}}}} - 1) [{\sigma^2} + {\rm{Tr}}\left( {{\bm{\Upsilon}} _c^{'}({{\mathbf{g}}_c}){ {\mathbf{Q}}_t}} \right)],\\
    		&\label{PasBF_Pro_Ori_amplitude} \quad \quad {\left[ {{{\mathbf{Q}}_t}} \right]_{(m,m)}} + {\left[ {{{\mathbf{Q}}_r}} \right]_{(m,m)}} = 1, \\
    		&\label{PasBF_Pro_Ori_rank} \quad \quad {\rm{rank}} \left( {{\mathbf{Q}}_t} \right) = 1,  {\rm{rank}} \left( {{\mathbf{Q}}_r} \right) =1,\\
    		& \quad \quad \eqref{U_initial}, 
    	\end{align}
    \end{subequations}where ${\bm{\Upsilon}}\left( {{\mathbf{g}}_c} \right) = {\rm{diag}} \left( {\mathbf{g}}_{c}^{H} \right) {{\mathbf{H}}_{{\rm{BR}}}} {{\mathbf{R}}_x} {\mathbf{H}}_{{\rm{BR}}}^{H} {\rm{diag }}\left( {\mathbf{g}}_{c} \right) \in {\mathbb{C}}^{ M_{t} \times M_{t} }$, ${\bm{\Upsilon}}^{'} \left( {{\mathbf{g}}_c} \right) =  {\rm{diag}} \left( {\mathbf{g}}_{c}^{H} \right) {{\mathbf{H}}_{{\rm{BR}}}} {{\mathbf{R}}_{s0}} {\mathbf{H}}_{{\rm{BR}}}^{H} {\rm{diag }}\left( {\mathbf{g}}_{c} \right) \in {\mathbb{C}}^{{M} \times {M} }.$ Let ${ {\bm{\Upsilon}}_s}( {\bm{\alpha}}_t) \buildrel \Delta \over = {\rm{diag}}( {\bm{\alpha}} _t^H){{\mathbf{H}}_{{\rm{BR}}}}{{\mathbf{R}}_x}{\mathbf{H}}_{{\rm{BR}}}^{H}{\rm{diag}}({\bm{\alpha}}_t) \in {\mathbb{C}}^{ M_{t} \times M_{t}} $, and ${\bm{\Upsilon}}^{'}({ {\mathbf{v}}_{t,2}} \odot {{\bm{\alpha}}_t}) \buildrel \Delta \over =  {\rm{diag}}\left( {{{\mathbf{v}}_{t,2}} \odot {\bm{\alpha}}_t^H} \right) {{\mathbf{H}}_{{\rm{BR}}}} {{\mathbf{R}}_x}{\mathbf{H}}_{{\rm{BR}}}^{H} {\rm{diag}}\left( {{\bm{\alpha}}_t} \right) \in {\mathbb{C}}^{M_{t} \times M_{t} }$, the expressions of ${{ \left[ {\mathbf{J}}_{ {{\tilde{{\bm{\eta}}}}}} \right] }_{(i,i)}}, i \in \left\{1,2\right\}$ are shown in \eqref{J_expression_pasBF}.
    \begin{figure*}
    	\begin{subequations}\label{J_expression_pasBF}
    		\begin{equation}
    			{{ \left[ {\mathbf{J}}_{ {\tilde{\bm{\theta}}}_{z} {{\tilde{\bm{\theta}} }_z} }  \right] }_{(1,1)}}= {{{\rm{\tilde{D}}}}_0} \{ ({M_r} + 1)({M_r} + 2)d_s^2\Re \left[ {{\rm{Tr}}({{\bm{\Upsilon}}_s}\left( {\bm{\alpha}}_t \right)  {{\mathbf{Q}}_r})} \right] + 12 d_R^2\Re \left[ {{\rm{Tr}}({ {\bm{\Upsilon}}_s}( {{\mathbf{v}}_{t,2}} \odot {{\bm{\alpha}}_t}){ {\mathbf{Q}}_r})} \right]\}, 
    		\end{equation}
    	    \begin{equation}
    	    	\begin{split}
    	    		{{ \left[ {\mathbf{J}}_{ {\tilde{\bm{\theta}}}_{z} {{\tilde{\bm{\theta}} }_z} }  \right] }_{(2,2)}}&= {{{\rm{\tilde{C}}}}_0} \Re \left[ {{\rm{Tr}}({{\bm{\Upsilon }}_s}( {\bm{\alpha}}_t) {{\mathbf{Q}}_r})} \right] + {{{\rm{\tilde{C}}}}_1} \Re \left[ {\rm{Tr}} \left( {{{\bm{\Upsilon}}_s}( {{\mathbf{v}}_{t,2}} \odot {{\bm{\alpha }}_t}){{\mathbf{Q}}_r}} \right)  \right]  \\
    	    		&- {{{\rm{\tilde{C}}}}_2} \Re \left[ {\rm{Tr}} \left({{\bm{\Upsilon}}_s^{'}({{\mathbf{v}}_{t,2}} \odot {{\bm{\alpha }}_t}){{\mathbf{Q}}_r}} \right) \right].  
    	    	\end{split}
    	    \end{equation}
    	\end{subequations}
       \hrulefill
    \end{figure*}Due to the rank one constraint \eqref{PasBF_Pro_Ori_rank}, the problem \eqref{PasBF_Pro_Ori} is still a non-convex problem, which is hard to solve. To deal with this dilemma, we take the equivalence transformation of \eqref{PasBF_Pro_Ori_rank} as $ \left\|\mathbf{Q}_{l}\right\|_* -\left\|\mathbf{Q}_{l}\right\|_2=0,  l \in\{t, r\}$, where $\left\|\mathbf{Q}_{l}\right\|_*=\sum_i {\sigma}_{i}\left(\mathbf{Q}_{l}\right)$ and $ \left\|\mathbf{Q}_{l}\right\|_2={\sigma }_{1} \left(\mathbf{Q}_{l} \right)$ denote the nuclear and spectral norm of ${\mathbf{Q}}_{l}$, respectively. $\sigma_{i}$ denotes the $i$th largest singular value of matrix ${\mathbf{{Q}}}$. Considering $ \forall \; {\mathbf{Q}}_{l} \in {\mathbb{H}}^{M}, {\mathbf{Q}}_{l} \succeq 0 $, the inequality $ \left\|\mathbf{Q}_{l}\right\|_* -\left\|\mathbf{Q}_{l}\right\|_2 \ge 0 $ always holds and the equation can be reached if and only if ${\rm{rank}} \left( {\mathbf{Q}}_{l} \right)=1 $. To evaluate the constraints violation, the penalty term can be expressed as $ \frac{1}{\rho} \sum\limits_{l \in {\left\{r,t\right\}}} \left( \left\|\mathbf{Q}_{l}\right\|_*- \left\|\mathbf{Q}_{l}\right\|_2 \right)  $, where $\rho$ denotes the penalty factor. By taking the penalty term as the objective function, the problem \eqref{PasBF_Pro_Ori} can be expressed as follows:
    \begin{subequations}\label{PasBF_Pro_Ori_penalty}
    	\begin{align}
    		& \mathop {\min} \limits_{ \left\{ {{{\mathbf{Q}}_r}}, {{\mathbf{Q}}_{t}} \right\}}  \;  \frac{1}{\rho} \sum\limits_{l \in {\left\{r,t\right\}}} \left( \left\|\mathbf{Q}_{l}\right\|_*- \left\|\mathbf{Q}_{l}\right\|_2 \right)  \\
    		& \quad {\rm{s.t.}} \; \quad  \eqref{U_initial}, \eqref{PasBF_Pro_Ori_SINR}-\eqref{PasBF_Pro_Ori_amplitude}.  
    	\end{align}
    \end{subequations}where $ {\rho}> 0$ decreases among consecutively iterations to increase the value of the objective function when the rank one solution cannot be guaranteed. To deal with the non-convexity of the penalty term, the SCA approach is adopted to derive a convex upper bound of the penalty term, which is given by
   \begin{algorithm}[!t]
   	\caption{The penalty-based joint BF Algorithm}\label{Alg_BF}
   	\begin{algorithmic}[1]\label{alg_2}
   		\STATE {\textbf{Input:}} 
   		\STATE {\textbf{Initialize:}} The feasible points $ \left\{  {\mathbf{{q}}}_{r}^{\left( 0 \right)}, {\mathbf{{q}}}_{t}^{\left( 0 \right)},  {\mathbf{v}}_{c}^{\left(0\right)}, {\mathbf{R}}_{s0}^{\left(0\right)} \right\}$, ${\rho}^{(0)}\ge 0$, and $0 <c <1$.  
   		\STATE The out layer iteration index $i=1$. 		 
   		\REPEAT 
   		\STATE The inner layer iteration index $n=1$.
   		\REPEAT 
   		\STATE Update $\left\{ {\mathbf{U}}^{\left(n\right)}, {\mathbf{V}}_{c}^{\left(n\right)}, {\mathbf{R}}_{s0}^{\left(n\right)} \right\}$ by solving \eqref{ActBF_Pro_Ori_Maxtrix_SDR}.
   		\STATE Update the penalty parts by \eqref{penalty_update}.
   		\STATE Update $\left\{ {\mathbf{Q}}_{r}^{\left(n\right)}, {\mathbf{Q}}_{t}^{\left(n\right)} \right\} $ by solving \eqref{PasBF_Pro_Ori_penalty_final}.
   		\UNTIL The relevant variation of objective value for \eqref{PasBF_Pro_Ori} among consecutive iterations is below a predefined threshold ${\epsilon}$ or the maximum number of inner iterations is reached. 
   		\STATE ${\rho}^{\left(i+1\right)}=c{\rho}^{\left(i\right)}.$
   		\STATE Reconstruct the diagonal matrix solution $\left\{ {\bm{\Theta}}_{r}^{\left(n\right)}, {\bm{\Theta}}_{t}^{\left(n\right)} \right\} $ via SVD.
   		\UNTIL{The relevant variation of objective value for \eqref{ActBF_Pro_Ori} among consecutive iteration is below a predefined threshold ${\epsilon}_{1}$.  }
   		\STATE \textbf{Output: \; $\left\{ {\mathbf{V}}_{c}^{*}, {\mathbf{R}}_{s0}^{*}, {\bm{\Theta}}_{r}^{*}, {\bm{\Theta}}_{t}^{*} \right\}$}. 
   	\end{algorithmic}
   \end{algorithm}
   
   \begin{algorithm}[tp]
   	\caption{AO algorithm for problem~\eqref{Eq:PF_ori}}\label{Alg_Overall}
   	\begin{algorithmic}[1]\label{alg_3}
   		\STATE Initialize feasible points $   \left\{ {M}_{r}^{\left(0\right)},d_{s}^{\left(0\right)} {\bm{\Theta}}_{r}^{\left( 0 \right)}, {\bm{\Theta}}_{t}^{\left( 0 \right)},  {\mathbf{V}}_{c}^{\left(0\right)}, {\mathbf{R}}_{s0}^{\left(0\right)} \right\}$ and the out layer iteration index $i=1$. 		
   		\REPEAT 
   		\STATE Update $\left\{ {M}_{r}^{\left(i\right)},d_{s}^{\left(i\right)}\right\} $ by \textbf{Algorithm~\ref{Alg_IntervalSub}}.
   		\STATE Update $\left\{ {\bm{\Theta}}_{r}^{\left( i \right)}, {\bm{\Theta}}_{t}^{\left( i \right)},  {\mathbf{V}}_{c}^{\left(i\right)}, {\mathbf{R}}_{s0}^{\left(i\right)}\right\} $ by \textbf{Algorithm~\ref{Alg_BF}}.
   		\UNTIL{The relevant variation of objective value for \eqref{Eq:PF_ori} among consecutive iteration is below a predefined threshold ${\epsilon}$.  }
   		\STATE \textbf{Output: \; $\left\{ {M}_{r}^{*},d_{s}^{*}, {\mathbf{v}}_{c}^{*}, {\mathbf{R}}_{x}^{*}, {\mathbf{q}}_{r}^{*}, {\mathbf{q}}_{t}^{*} \right\}$}. 
   	\end{algorithmic}
   \end{algorithm}

   \begin{equation}\label{penalty_update}
   	\left\|\mathbf{Q}_{l} \right\|_*-\left\|\mathbf{Q}_{l}\right\|_2 \leq \left\|\mathbf{Q}_{l}\right\|_*-\overline{\mathbf{Q}}_{l}^{n},
   \end{equation}where $\overline{\mathbf{Q}}_{l}^{n} \quad \triangleq \quad\left\|\mathbf{Q}_{l}^{n}\right\|_2+{\rm{Tr }} \left[\overline{\mathbf{u}} \left(\mathbf{Q}_{l}^{n}\right) \left(\overline{\mathbf{u}}\left(\mathbf{Q}_{l}^{n}\right) \right)^H\left(\mathbf{Q}_{l}-\mathbf{Q}_{l}^{n}\right)\right] $, $\overline{\mathbf{u}} \left(\mathbf{Q}_{l}^{n}\right)$ is the largest eigenvector corresponding to ${\mathbf{Q}}_{l}$. Then, problem \eqref{PasBF_Pro_Ori_penalty} can be transformed into a convex optimization as follows: 
   \begin{subequations}\label{PasBF_Pro_Ori_penalty_final}
   	\begin{align}
   		& \mathop {\min} \limits_{ \left\{ {{{\mathbf{Q}}_r}}, {{\mathbf{Q}}_{t}} \right\}}  \;  \frac{1}{\rho} \sum\limits_{l \in {\left\{r,t\right\}}} \left( \left\|\mathbf{Q}_{l}\right\|_*-\overline{\mathbf{Q}}_{l}^{n}  \right)  \\
   		& \quad {\rm{s.t.}} \; \quad  \eqref{U_initial}, \eqref{PasBF_Pro_Ori_SINR}-\eqref{PasBF_Pro_Ori_amplitude},
   	\end{align}
   \end{subequations}which can be solved via CVX~\cite{cvx}. To sum up, the penalty-based two-layer iterative algorithm is proposed as shown in~\textbf{Algorithm~\ref{Alg_BF}}. For the inner layer, the convex optimization for active BF $ \left\{ {\mathbf{V}}_{c}, {\mathbf{R}}_{s0}\right\}$ and passive BF are consecutive solved until the termination condition is satisfied as $ \sum_{l \in {\left\{r,t\right\}}} \left( \left\|\mathbf{Q}_{l}\right\|_*-\overline{\mathbf{Q}}_{l}^{n }\right) \le {\epsilon}.$ For the outer layer, the penalty factor $\rho$ is updated as ${\rho}^{\left(i+1\right)}=c{\rho}^{\left(i\right)}, c<1 $ to guarantee the rank one solution. To guarantee the convergence of~\textbf{Algorithm~\ref{Alg_Overall}}, the penalty term is set as a large value at the initial which decreased with the increased iteration. In a nutshell, the overall algorithm to optimize the sensor deployment and BF design is summarised in~\textbf{Algorithm~\ref{Alg_Overall}}. The proposed sensor deployment and beamforming design algorithms are alternatively optimized until the convergence has been achieved.  
    
\subsection{Overall Algorithm: Convergence and Complexity}
\begin{figure}[!t]
	\centering
	\includegraphics[width=3in]{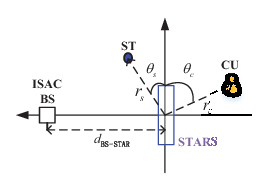}
	\caption{The setup of the STARS-assisted NF ISAC system}
\end{figure}
The convergence performance of \textbf{Algorithm~\ref{Alg_Overall}} is summarized in the following. It is noted that the recovery for discrete variables with the comparison of adjacent integers can generally guarantee the convergence of the proposed joint sensor deployment algorithm in practice. For the convergence of \textbf{Algorithm~\ref{Alg_BF}}, the objective function of \eqref{ActBF_Pro_Ori} is non-increasing among each iteration in the inner layer. With the decreased penalty parts among the consecutive out layer iteration, \textbf{Algorithm~\ref{Alg_BF}} converges to a Karush-Kuhn-Tucker point of \eqref{ActBF_Pro_Ori} when the penalty parts approach zero~\cite{Qingjiang_20TSP_1}. During the alternative optimization, the objective value is non-increasing within one iteration. Considering the objective function is lower bounded by the limited transmit power and aperture constraints, the proposed \textbf{Algorithm~\ref{Alg_Overall}} is guaranteed to converge to a stationary point of~\eqref{Eq:PF_ori}. The complexity of \textbf{Algorithm~\ref{Alg_Overall}} is analysed as follows. The main complexity of \textbf{Algorithm~\ref{Alg_Overall}} arise from \textbf{Algorithm~\ref{Alg_BF}} due to the high dimension of $\left\{ {\mathbf{\Theta}}_{l}, {\mathbf{V}}_{c}, {\mathbf{R}}_{s0} \right\}, l \in \left\{r,t\right\}$. By employing the inter-point method, the complexity of solving \eqref{ActBF_Pro_Ori_Maxtrix_SDR} and \eqref{PasBF_Pro_Ori_penalty_final} are given by ${\mathcal{O}}\left( \left(N^{3.5}+2^{3.5} \right) \log \left( \frac{1}{\epsilon}\right)\right) $ and ${\mathcal{O}}\left( \left(M^{3.5}2^{3.5} \right)\log \left( \frac{1}{\epsilon}\right) \right)$, respectively. Let $I_{2,i}$ and $I_{2,o}$ denote the inner layer and out layer iteration number, the complexity of \textbf{Algorithm~\ref{Alg_BF}} is given by ${\mathcal{O}}\left( I_{2,o} I_{2,i}\left(N^{3.5}+2^{3.5}+2^{3.5}M^{3.5} \right) \log \left( \frac{1}{\epsilon}\right) \right)$. By assuming $ I_{3}$ denotes the iteration number of \textbf{Algorithm~\ref{Alg_Overall}}, the complexity of \textbf{Algorithm~\ref{Alg_Overall}} is in the order of ${\mathcal{O}}\left( I_{3} I_{2,o} I_{2,i}\left(N^{3.5}+2^{3.5}+2^{3.5}M^{3.5} \right) \log \left( \frac{1}{\epsilon}\right) \right)$.

\section{Numerical Results}
\begin{figure}[!t]
	\centering
	\includegraphics[width=3.5in]{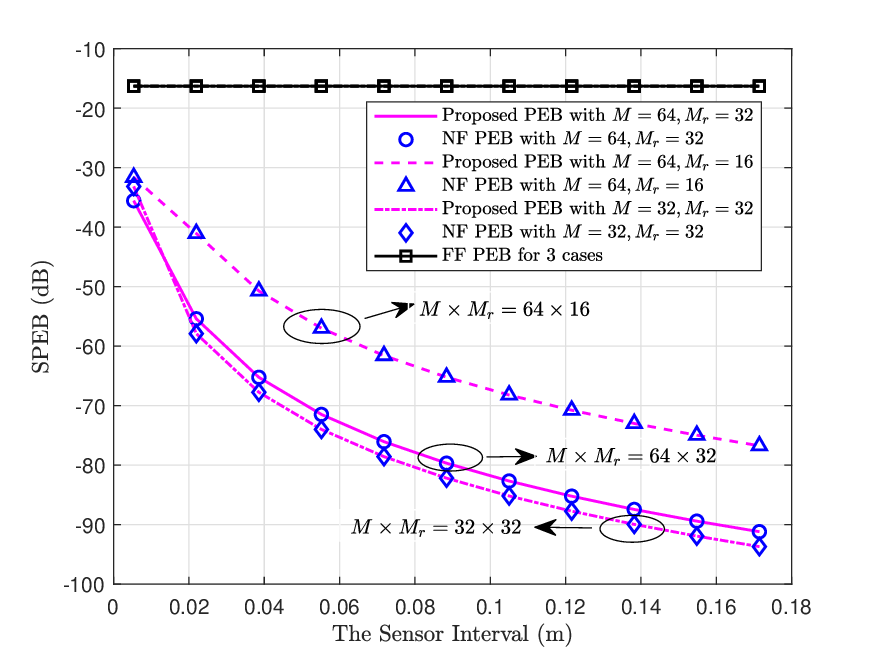}\\
	\caption{The impact of sensor interval and sensor numbers on SPEB.}\label{Fig_Val_1}
\end{figure}
\begin{figure}[!t]
	\centering
	\includegraphics[width=3.5in]{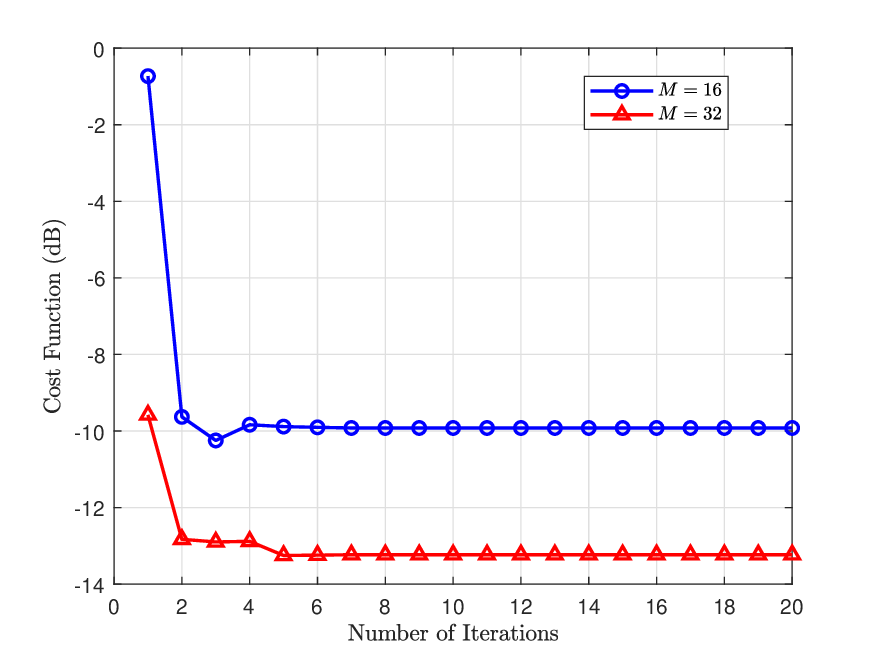}\\
	\caption{The convergence performance of \textbf{Algorithm~\ref{Alg_Overall}}.}\label{Convergence}
\end{figure}

In this section, numerical results are provided via Monte Carlo simulations to verify the benefits of the proposed algorithms under the NF propagation surroundings. In particular, the BS equipped with a ULA $N=128$ antennas is assumed to operate at ${f}_{c}=28$GHz $({\lambda}_{c}=0.0107m)$. The adjacent antennas maintain a half-wavelength interval and the antenna aperture at the BS is ${D}_{B}=\frac{\left(N-1\right) {\lambda}_{c}}{2} \approx 0.68m$. The STARS is assumed to equip with a linear uniformly distributed $M=32$ STAR elements with the spacing $d_{R}=\frac{\lambda_{c}}{2}$, where the aperture of STAR elements is ${D}_{R}=\frac{ \left(M-1\right) \lambda_{c}}{2}\approx 0.17m$. There is one sensing target and one communication user in each half-space. During each channel realization, the CU and ST are randomly located among the STARS where its angles are randomly distributed in $\left[45^{\circ},135^{\circ}\right]$. The weight of normalized SPEB is set as ${w}_{0}=0.5$ and the most affordable number of sensor elements is set as ${M}_{0}=M$. The predefined sensing precision is set as ${\epsilon}_{0}=10^{-5}$. Without loss of generality, the algorithm initialization is set to ${\mathbf{R}}_{x}^{0}=\frac{P_{max}}{N} {\mathbf{I}}_{N}, {\bm{\Theta}}_{r}^{0}=\frac{\sqrt{2}}{2}{\mathbf{I}}_{N}, {M}_{r}^{0}=M$, and ${d}_{s}^{0}=0.5{\lambda}_{c}$. We set the convergence parameters ${\epsilon}=10^{-5}$ and ${\epsilon}_{1}=10^{-8}$. To evaluate the performance of the proposed algorithms, three benchmark schemes are considered as follows: 
\vspace{-0.1cm}
\begin{itemize}
	\item \textbf{Random Sensor Deployment (RandM)}: In this scheme, the interval among two radio sensors is fixed as $d_{{\rm{R}}}=0.5{\lambda}$. The number of radio sensors is randomly distributed among $\left[2, {M}\right]$. The active and passive BF designs are optimized via \textbf{Algorithm~\ref{Alg_BF}}. 
	
	\item \textbf{Conventional RIS (conRIS)}: In this scheme, the full space radio sensing and communication functionalities are implemented via one reflecting-only RIS and one transmitting-only RIS. For the fair comparison, there are $\frac{ M_{t} }{2}$ elements in each separate RIS, where $M_{t}$ is assumed to be an even number for simplicity. Thus, the corresponding optimization can be solved via \textbf{Algorithm ~\ref{Alg_BF}}. 
	
	\item \textbf{Conventional ISAC (conISAC)}: In this scheme, the communication and radio sensing functionalities are both conducted via the joint C\&S beams as shown in~\cite{FanLiu_22TSP}. We adopt the block coordinate descent (BCD) approach where the active BF is solved via SDR and the passive BF is solved via \textbf{Algorithm~\ref{Alg_BF}}. 
\end{itemize}
\vspace{-0.4cm}
\subsection{Validation of the Derived SPEB Expression} 
\begin{figure}[!t]
	\centering
	\subfigure[The cost function value versus ${P}_{max}$.]{\label{CF_Pmax}
		\includegraphics[scale=0.5]{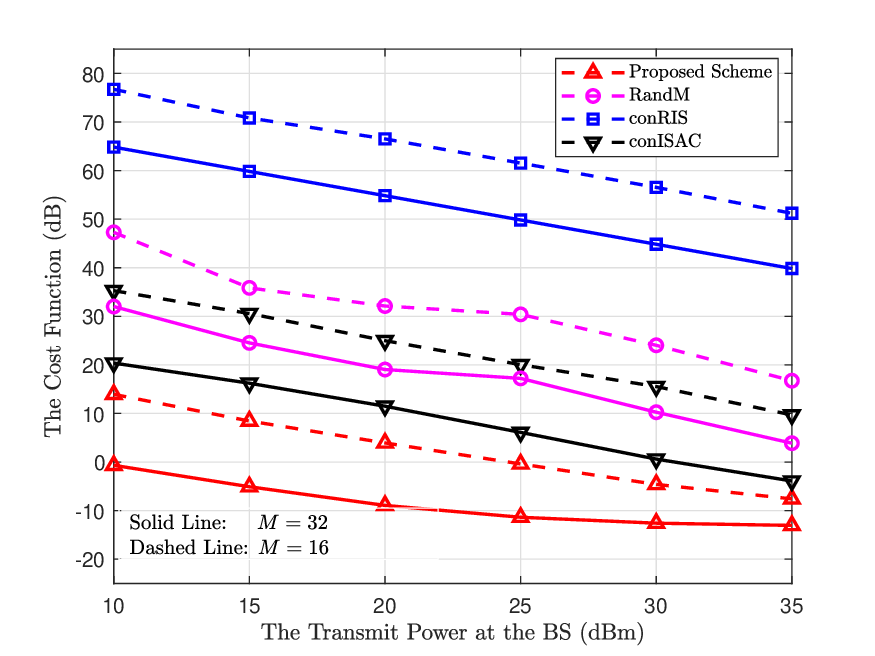}}
	\subfigure[The SPEB versus ${P}_{max}$.]{\label{PEB_Pmax}
		\includegraphics[scale=0.5]{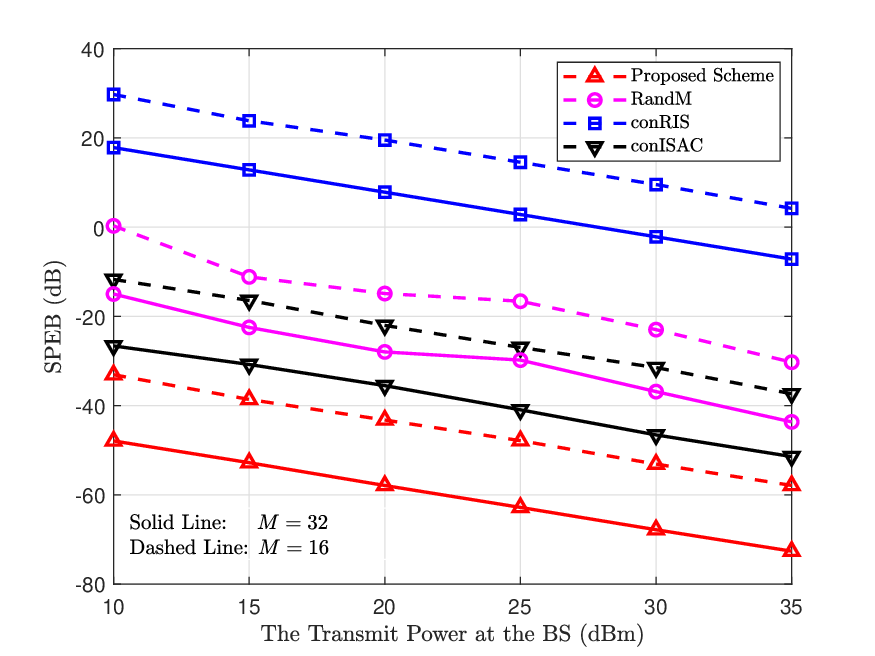} }
	\caption{The performance comparison versus ${P}_{max}$.}
	\label{PEBCF_Pmax}
\end{figure}
To validate the quantity relationship between the SPEB and the sensor deployment, we compare the SPEB performance with the numerical SPEB~\cite{Zhaolin23_NF}, which has no approximation for the NF steering vector. With the fixed active and passive BF, the impact of the sensor number on SPEB has been explored in Fig.~\ref{Fig_Val_1} with different numbers of STAR and sensor elements. From Fig.~\ref{Fig_Val_1}, it is clear that our proposed SPEB derivation is close to the numerical SPEB. Compared with the FF estimation, the NF radio sensing estimation is much more sensitive to the sensor interval and the number of sensor elements, where the SPEB decreases with the increase of the sensor interval. This is because the FF estimation employs the planar wavefronts to model the received signal at the phased array, which mainly depends on the aperture size of the antenna array. As the NF estimation employs spherical wavefronts to depict the received signal, the SPEB performance is decided by the number of sensor elements and the sensor interval, which validate the importance of our work. Compared with the NF radio sensing among three cases (i.e., ${M} \times {M}_{r}=64 \times 32$, ${M} \times {M}_{r}=32 \times 32 $, and ${M} \times {M}_{r}=64 \times 16$), it is clear that the SPEB performance is always limited by the minimum number among the STAR and sensor elements. 
\vspace{-0.3cm}
\subsection{Convergence Performance}
In Fig.~\ref{Convergence}, we compare the out-layer convergence behavior of the proposed AO with different STARS elements $M_{t}$. We set $N =128$, ${R}_{min}=0$dB, and ${M}_{0} = M_{t}$. The STARS is equipped with ${M}=16$ and ${M}=32$, respectively. The results are obtained via one channel realization. As shown in Fig. 4, the CF of \textbf{Algorithm~\ref{Alg_Overall}} decreases quickly with the outer layer iteration increased and converge at around 6 iterations for both scenarios, which validates the effectiveness of our proposed algorithms.
\vspace{-0.2cm}
\subsection{Radio Sensing Performance}

\subsubsection{The Radio Sensing Performance Versus Different ${P}_{max}$}

\begin{figure}[!t]
	\centering
	\subfigure[The cost function value versus $M$.]{\label{CF_M}
		\includegraphics[scale=0.5]{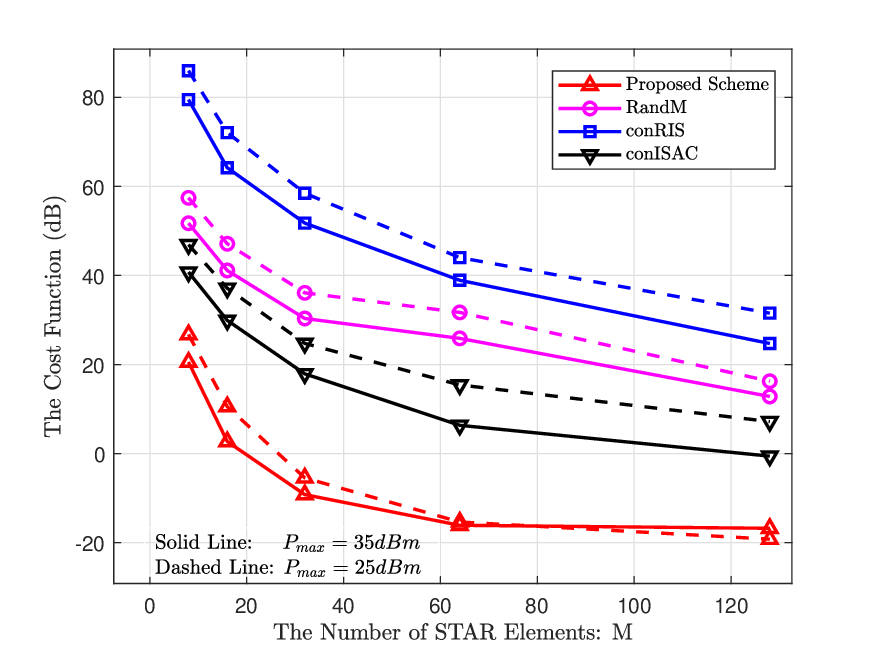}}
	\subfigure[The SPEB versus $M$.]{\label{PEB_M}
		\includegraphics[scale=0.5]{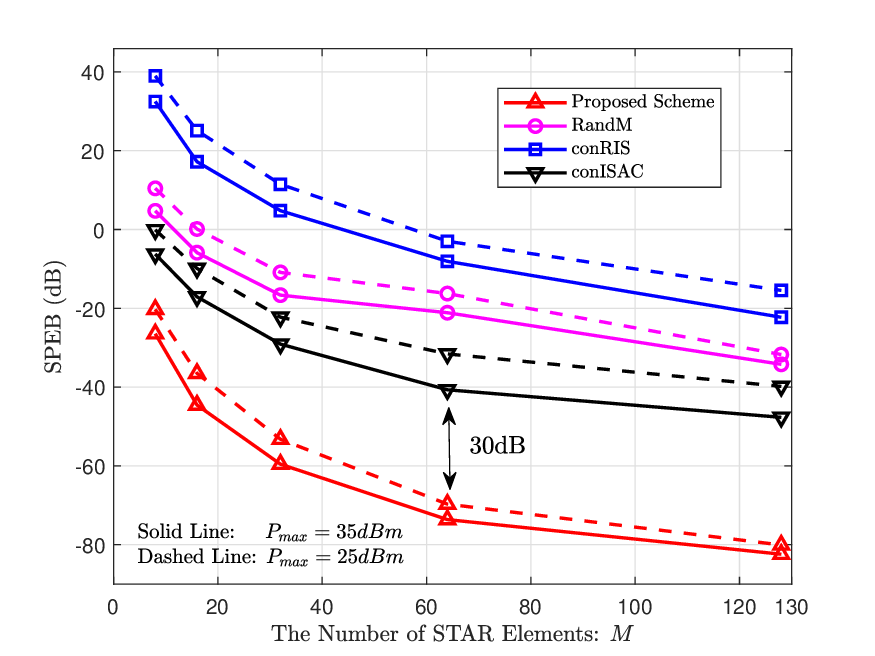} }
	\caption{The performance comparison versus ${M}$.}
	\label{PEBCF_M}
\end{figure}
In Fig.~\ref{PEBCF_Pmax}, we compare the CF and SPEB performance versus the transmit power $P_{max}$ at the BS for $M=32$ and $M=16$. We set ${\bar{R}}_{min}=0$dB and ${\epsilon_{0}}=10^{-5}$. The STs are randomly distributed around  ${\mathbf{p}}_{s}=\left[0,20,0\right]^{{\rm{T}}}$ with ${r}_{sen}=3$. It can be seen that our proposed algorithms achieve the least SPEB and CF compared with the other schemes. The reason can be explained as follows. First, the proposed scheme improves the radio sensing performance with the aid of more STAR elements compared with the conRIS scheme. Second, the dedicated sensing beam in our proposed algorithm is prone to focus on the ST compared with the conISAC scheme. Thirdly, our proposed algorithms save the deployment cost to achieve a better radio sensing performance as well as a lower CF compared with the RandM scheme. Meanwhile, the SPEB and CF are both decreased with the increased ${P}_{max}$. This is because the strength of the received echo signal depends on the transmit power, where more transmission power enables the stronger signal strength to keep a less estimation error. 

\subsubsection{The Radio Sensing Performance Versus Different $M$}
In Fig.~\ref{PEBCF_M}, we compare the radio sensing performance versus the number of STAR elements with $P_{max}=25$dBm and $P_{max}=35$dBm. We set $N=128$, ${\bar{R}}_{min} = 0$dB and ${\epsilon}_{0}=10^{-5}$. The STs are randomly distributed around ${\mathbf{p}}_{s}=\left[0, 10,0 \right]^{{\rm{T}}}$ with ${r}_{sen}=10$. The proposed algorithms achieve the least CF and SPEB compared with the other schemes, which validates the superiority of employing the dedicated sensing beam. The CF and SPEB of all four algorithms decrease with the increased STAR elements in Fig~\ref{PEBCF_M}. This is because more STAR elements facilitate a greater beampattern gain at the ST benefiting from the focus on both the distance and angle domain via NF propagation. Compared with the RandM scheme, our proposed algorithm achieves at least 20dB gain via the optimization of sensor deployment, which strengthens the necessity of our work. Meanwhile, the gap between our proposed scheme and the benchmark schemes enlarged with the increased STAR number. When the number of STAR elements is greater than $64$, the proposed algorithms with two scenarios reach the close performance in Fig~\ref{CF_M}, which shows the superiority in larger scale STARS with limited energy consumption. The proposed algorithm reaches at least $30$dB gain compared with the other benchmark schemes if $M \ge 64$ in Fig.~\ref{PEB_M}. 

\subsubsection{The Radio Sensing Performance Comparison With ML and MUSIC}
\begin{figure}[!t]
	\centering
	\includegraphics[width=3in]{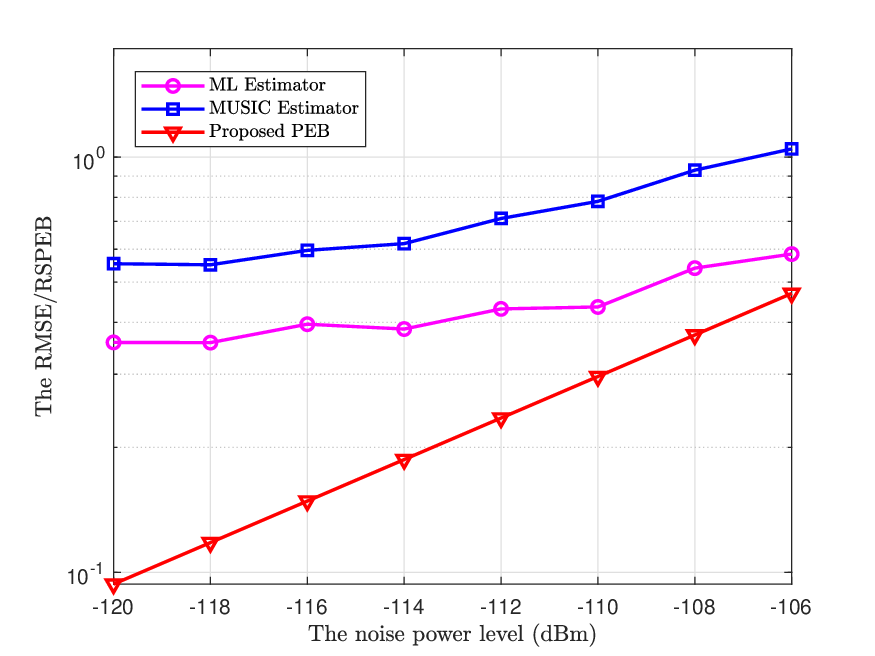}\\
	\caption{The RMSE/RSPEB for the position estimation versus noise power at the sensors.}\label{RMSE}
\end{figure}
In Fig.~\ref{RMSE}, we compare the root mean square error (RMSE)/root SPEB between the proposed SPEB approximation and the classic ML and MUSIC estimation algorithm. We set $P_{max}=30$dBm, $N=64$, $M={M}_{r}=16$. The intervals between the adjacent STAR and sensor elements are set to $\frac{\lambda}{2}$. The RMSE/RSPEB increases with the increase of the noise power. With the increase of the noise power, the received power of the useful echo signal decreased, which makes it harder to distinguish the useful signal and the noise. Meanwhile, the proposed SPEB approximation is the lower bound of the classical ML and MUSIC estimation algorithm. The gap between the proposed SPEB and ML/MUSIC decreased with the smaller estimation resolution. As the steering vectors of the same angle among different distances in NF propagation cannot hold the orthogonal, the distance estimation with ML and MUSIC algorithm cannot reach a close performance to ${\rm{CRB}}_{r}$, where the classical estimation algorithm should be modified to adapt to the NF propagation.

\section{Conclusion} 
In this paper, a STARS-enabled NF ISAC framework has been investigated, where the `sensing-at-STARS' architecture was adopted to combat the signal attenuation due to multi-hop sensing links. Benefiting from the spherical wavefronts model in NF propagation environment, the estimation on distance was invoked to further elevate the radio sensing performance. A new expression of SPEB was derived to quantitatively reveal the dependence of the joint BF design and sensor deployment, where the radio sensing performance is improved with the increment in sensor numbers and/or sensor interval. A conditional optimal solution for sensor interval was given while the SCA and GP were employed to optimize the number of radio sensors. To solve the highly-coupled joint BF design problem, a penalty-based two-layer iterative algorithm was proposed by exploiting the Schur complement, SDR and SCA technologies. Numerical results validated that the derived SPEB expression has a good approximation to the exact SPEB, where the FIM of the estimated position in the NF propagation can be approximated by its diagonal matrix. With the same joint BF design, the best SPEB performance will be reached if the aperture size of the STAR elements is equal to the sensor elements. Meanwhile, our proposed joint sensor deployment and BF algorithm achieved the best radio sensing performance accompanied by the lowest deployment cost compared to the benchmarks.

\appendix

\begin{figure*}
	\begin{equation}\label{FIM_element_Joint}
		\begin{split}
			{\mathbf{J}}_{{{\tilde{{\bm{\eta}}}}} } &= \frac{{2{{\left| {{\alpha_s}} \right|}^2}L}} {{{\sigma ^2}}} [ \begin{array}{*{20}{c}}
				{\begin{array}{*{20}{c}}
						{\Re [{\rm{Tr}} ( {{\mathbf{R}}_{ {\bar{{\mathbf{X}}}}_{r} }}  {\dot{\mathbf{A}}}_{{\theta _z}}^H{{\dot{\mathbf{A}}}_{{\theta_z}}})] }\\
						{\Re [{\rm{Tr}} ( {{\mathbf{R}}_{ {\bar{{\mathbf{X}}}}_{r} }}  {\dot{\mathbf{A}}}_{{r _z}}^H{{\dot{\mathbf{A}}}_{{\theta_z}}})]}
				\end{array}}&{\begin{array}{*{20}{c}}
						{\Re [{\rm{Tr}} ( {{\mathbf{R}}_{ {\bar{{\mathbf{X}}}}_{r} }} {\dot{\mathbf{A}}}_{{\theta_z}}^H{{\dot{\mathbf{A}}}_{{r_z}}})] }\\
						{\Re [{\rm{Tr}} ( {{\mathbf{R}}_{ {\bar{{\mathbf{X}}}}_{r} }}  {\dot{\mathbf{A}}}_{{r_z}}^H{{\dot{\mathbf{A}}}_{{r_z}}})] }, 
				\end{array}}
			\end{array} ]
		\end{split}
	\end{equation} 
	\hrulefill
\end{figure*}

\subsection{Proof of {\textbf{Proposition~\ref{Pro:FIM_expression}}} } \label{App:FIM_expression}
\renewcommand{\theequation}{A.\arabic{equation}}
\setcounter{equation}{0}

Let ${\mathbf{R}}_{ {\bar{{\mathbf{X}}}}_{r}}=\frac{ {\bar{{\mathbf{X}}}}_{r} {\bar{{\mathbf{X}}}}_{r}^{H} }{L} $, where ${\bar{{\mathbf{X}}}}_{r} =\left[ {\bar{{\mathbf{x}}}}_{r}\left[1\right], ..., {\bar{{\mathbf{x}}}}_{r}\left[l\right]\right] \in {\mathbb{C}}^{M \times L}$, the entries of ${\mathbf{J}}_{{{\tilde{{\bm{\theta}}}}_s} {{\tilde{{\bm{\theta}}}}_s}} $ are expressed as~\eqref{FIM_element_Joint}, where ${\bar{{\mathbf{A}}}}_{t}={{\bm{\alpha}}_t} {\bm{\alpha}}_t^H \in {\mathbb{C}}^{M \times M}$. As the entries are determined by $ \left\{ {\dot{\mathbf{A}}}_{{\theta _z}}^H{{\dot{\mathbf{A}}}_{{\theta_z}}}, {\dot{\mathbf{A}}}_{{\theta _z}}^H{{\dot{\mathbf{A}}}_{{r_z}}}, {\dot{\mathbf{A}}}_{{r_z}}^H {{\dot{\mathbf{A}}}_{{r_z}}}\right\}$, we first exploit the symmetry property of ${\mathbf{v}}_{l,1}$ and the Hermitian property of $ {{\mathbf{R}}_{ {\bar{{\mathbf{X}}}}_{r} }}, {{\bar {\mathbf{A}}}_t}$, where we have the following findings: 
	\begin{itemize}\label{V_finding}
		\item 1):  ${\bm{\alpha}}_{r}^{H} {\rm{diag}}\left({\mathbf{v}}_{r,1}\right) {\rm{diag}} \left( {\mathbf{v}}_{r,1} \right) {\bm{\alpha}}_{r}=\frac{ M_{r} \left({M}_{r}+1\right) \left({M}_{r}+2\right) }{12} $; \\ 
		\item 2): $ {\rm{Tr}}( {\rm{diag}} \left( {{\mathbf{v}}_{t,1}} \right) {\mathbf{R}}_{{{\bar{\mathbf{X}}}_r}} {\bar {\mathbf{A}}_t} ) = 0 $; \\ 
		\item 3): ${{\mathbf{A}}^H}{\rm{diag}}({{\mathbf{v}}_{r,2}} \odot {{\mathbf{v}}_{r,2}}){\mathbf{A}}=\frac{{M_r({M_r} + 1)({M_r} + 2)(3M_r^2 + 6{M_r} - 4){{{\mathbf{\bar{A}}}}_t}}} {{240}} $.  
	\end{itemize}	The first finding can be proved via~\eqref{App3_Pro}, which can be derived as
	\begin{equation}
		\begin{split}
			& \quad {\bm{\alpha}}_{r}^{H} {\rm{diag}}\left({\mathbf{v}}_{r,1}\right) {\rm{diag}} \left( {\mathbf{v}}_{r,1} \right) {\bm{\alpha}}_{r} = 
			\left\| { {{{\mathbf{v}}_{r,1}}} } \right\|_{2}^{2} {\bm{\alpha}}_{r}^{H} {\bm{\alpha}}_{r} \\
			& = \left\| { {{{\mathbf{v}}_{r,1}}} } \right\|_{2}^{2} = 2 \sum\limits_{i = 1}^{{{\tilde M}_r}} {{i^2}} = \frac{ M_{r} \left({M}_{r}+1\right) \left({M}_{r}+2\right) }{12}. 
		\end{split}
	\end{equation}Similarly, the third finding can be derived as:
	\begin{equation}
		\begin{split}
			& {{\mathbf{A}}^H} {\rm{diag}}({{\mathbf{v}}_{r,2}} \odot {{\mathbf{v}}_{r,2}}) {\mathbf{A}}={\bm{\alpha}}_{t} {\bm{\alpha}}_{r}^{H} {\rm{diag}}({{\mathbf{v}}_{r,2}} \odot {{\mathbf{v}}_{r,2}}) {\bm{\alpha}}_{r} {\bm{\alpha}}_{t}^{H}  \\
			& =\left[{\bm{\alpha}}_{r}^{H} {\rm{diag}}({{\mathbf{v}}_{r,2}} \odot {{\mathbf{v}}_{r,2}}) {\bm{\alpha}}_{r} \right] {\bm{\alpha}}_{t}{\bm{\alpha}}_{t}^{H} \\
			& = \left\| { {{{\mathbf{v}}_{r,2}}} } \right\|_{2}^{2} {{{\mathbf{\bar{A}}}}_t} = 2 \sum\limits_{i = 1}^{{{\tilde M}_r}} {{i^4}} {{{\mathbf{\bar{A}}}}_t} \\
			& = \frac{ {M_r({M_r} + 1)({M_r} + 2)(3M_r^2 + 6{M_r} - 4) {{{\mathbf{\bar{A}}}}_t} }} {{240}}.
		\end{split}
	\end{equation}Due to the symmetry of ${\mathbf{v}}_{l,1}, \quad  l \in \left\{r,t\right\} $, we have ${\rm{Tr}}\left({\rm{diag}} \left({\mathbf{v}}_{l,1}\right) \right)= \sum\limits_{i = 1}^{{M_l}} {{{\left[ {{{\mathbf{v}}_{l,1}}} \right]}_i}}=0 $. Recalling that $\left| { {\rm{Tr}}({\mathbf{A}}{\mathbf{B}})} \right| \le {\rm{Tr}}({\mathbf{A}}) {\rm{Tr}}({\mathbf{B}}), \forall {\mathbf{A}} , {\mathbf{B}} \in {\mathbb{F}}^{n}  $~\cite{MatrixTheory_Trace_2}, it can be inferred that $ \left|{ {\rm{Tr}}\left( {{\mathbf{v}}_{t,1}}  {\mathbf{R}}_{{{\bar{\mathbf{X}}}_r}} {\bar{\mathbf{A}}_t} \right) }\right| \le {\rm{Tr}}\left( {{\mathbf{v}}_{t,1}} \right) {\rm{Tr}} \left( {\mathbf{R}}_{{{\bar{\mathbf{X}}}_r}} {\bar{\mathbf{A}}_t} \right)=0$, which indicates that ${\rm{Tr}}\left( {{\mathbf{v}}_{t,1}}  {\mathbf{R}}_{{{\bar{\mathbf{X}}}_r}} {\bar{\mathbf{A}}_t} \right) =0$. Therefore, $\left[ {\mathbf{J}}_{ {\tilde{{\bm{\theta}}}}_{s} {\tilde{{\bm{\theta}}}}_{s}} \right]_{\left(1,2\right)}= \left[ {\mathbf{J}}_{ {\tilde{{\bm{\theta}}}}_{s} {\tilde{{\bm{\theta}}}}_{s}} \right]_{\left(2,1 \right)}=0 $, where  ${\mathbf{J}}_{{{\tilde{{\bm{\theta}}}}_s}{{\tilde{{\bm{\theta}}}}_s}} $ degrades into a diagonal matrix. The ${\mathbf{J}}_{{{\tilde{{\bm{\theta}}}}_s}{{\tilde{{\bm{\theta}}}}_s}} ^{-1} $ is given by $ {\mathbf{J}}_{{{\tilde{{\bm{\theta}}}}_s}{{\tilde{{\bm{\theta}}}}_s}} ^{-1}={\rm{diag}} \left( \left[ \left[ {\mathbf{J}}_{{{\tilde{{\bm{\theta}}}}_s}{{\tilde{{\bm{\theta}}}}_s}} \right]_{\left(1,1\right)}^{-1} ; \left[ {\mathbf{J}}_{{{\tilde{{\bm{\theta}}}}_s}{{\tilde{{\bm{\theta}}}}_s}} \right]_{\left(2,2\right)}^{-1} \right] \right)$. By substituting the mentioned properties into \eqref{FIM_element_Joint}, the expression of $\left[{\mathbf{J}}_{{{\tilde{{\bm{\theta}}}}_s}{{\tilde{{\bm{\theta}}}}_s}} \right]_{\left(i,i\right)}, i \in \left\{1,2\right\} $ can be further expressed as~\eqref{J_expression}. The expression of SPEB can be recast as:
	
    \begin{small}
    	\begin{equation}
    		\begin{split}
    			& {\rm{SPEB}} \left(\{ {\mathbf{R}}_{x}, {\mathbf{q}}_{r}, {\bm{\Theta}}_{r} \} ; {\bm{\eta}} \right) = {\rm{tr}}\left( { \left[ {{ { ( {\mathbf{T}}  {\mathbf{J}}_{ {{\tilde{{\bm{\theta}}}}_s} {\tilde{{\bm{\theta}}}}_{s} }  {{\mathbf{T}}^H} )}^{-1}}} \right] } \right) \\
    			=& {\rm{tr}}\left( {\mathbf{J}}_{ {{\tilde{{\bm{\theta}}}}_s} {\tilde{{\bm{\theta}}}}_{s} }^{-1} { \left[ {{   {{\mathbf{T}}^H} {\mathbf{T}} }} \right]^{-1} } \right) = \frac{ {p}_{x}^{2} {p}_{y}^{2} \sum\limits_{i = 1}^2 \left\{  { [ {\mathbf{J}}_{ {{\tilde{{\bm{\theta}}}}_s} {\tilde{{\bm{\theta}}}}_{s} } ]}_{(i,i)}^{-1} { {[ {\tilde{{\mathbf{T}}}} ]}_{(i,i)} } \right\}  }{ \left({p}_{x}+{p}_{y}\right)^{2} } ,
    		\end{split}
    	\end{equation}
    \end{small}

 \noindent The proof is complete.

\subsection{Proof of {\textbf{Proposition~\ref{App_Opt_ds}}} } {\label{App:Opt_ds}}
\renewcommand{\theequation}{B.\arabic{equation}}
\setcounter{equation}{0}
Let $f\left( {\tilde{d}}_{s} \right)= f_{1} \left( {\tilde{d}}_{s} \right)+ f_{2} \left( {\tilde{d}}_{s} \right)$, where $f_{1} \left( {\tilde{d}}_{s} \right) =\frac{ {[\tilde{{\mathbf{T}}}]}_{(1,1)} } { { {\rm{C}}_{11}} {{\tilde{d}}_s}+{\rm{C}}_{10} } $, $f_{2} \left( {\tilde{d}}_{s} \right)=\frac{{{{[\tilde{{\mathbf{T}}} ]}_{(2,2)}}}} { {\rm{C}_2} {\tilde{d}}_s^2 - {\rm{C}_1} {{\tilde{d}}_s} + {{\rm{C}_0}} }=\frac{{{{[\tilde{{\mathbf{T}}} ]}_{(2,2)}}}}  { {\rm{C}_2}\left( {\tilde{d}}_s-\frac{{\rm{C}_{1}}} {2{\rm{C}}_{2}} \right)^{2} + {\rm{C}_0} - \frac{ {\rm{C}}_{1}^{2} } { 4{{\rm{C}}_{2} } } } $. As ${\rm{C}}_{11}>0 $, it's clear that $f_{1}\left({\tilde{d}}_{s}\right) $ is monotonically decreasing, which indicates $\left[ f_{1}\left( {\tilde{d}}_{s} \right) \right]_{min} =f_{1} \left( \frac{ {M}^{2} d_{{\rm{R}}}^{2} }{ {M}_{r}^{2} } \right) $. With the different relationship among ${{\rm{C}}_0} - \frac{ {\rm{C}}_{1}^{2} }{ 4{{\rm{C}}}_{2} } \& 0 $ and $ \left| {\frac{{{{\rm{C}}_1}}}{{2{{\rm{C}}_2}}} - \frac{ {M}^{2} d_{{\rm{R}}}^{2} }{4}  } \right| \& \left| { \frac{{{{\rm{C}}_1}}}{{2{{\rm{C}}_2}}} - \frac{ {M}^{2} d_{{\rm{R}}}^{2} }{ {M}_{r}^{2} } } \right| $, $ f_{2} \left( {\tilde{d}}_{s} \right)$ not maintains the monotonic decreasing property on $ {\tilde{d}}_{s}$. If $ {{\rm{C}}_0} - \frac{ {\rm{C}}_{1}^{2} }{ 4{{\rm{C}}}_{2} } >0$,  $f_{2} \left( {\tilde{d}}_{s} \right) >0, f\left( {\tilde{d}}_{s} \right)>f_{1} \left( {\tilde{d}}_{s} \right) \ge f_{1}\left( \frac{ {M}^{2} d_{{\rm{R}}}^{2} }{{M}_{r}^{2}} \right) , \forall {\tilde{d}}_{s} \in \left[\frac{ {M}^{2} d_{{\rm{R}}}^{2} }{ 4 }, \frac{ {M}^{2} d_{{\rm{R}}}^{2} }{ {M}_{r}^{2} } \right] $. Under this condition, the optimal value of $f_{2}\left( {\tilde{d}}_{s} \right)$ can also be reached at ${\tilde{d}}_{s}=\left(\frac{ {M}^{2} d_{{\rm{R}}}^{2} }{ {M}_{r}^{2} } \right) $ if $\left| {\frac{{{{\rm{C}}_1}}}{{2{{\rm{C}}_2}}} - \frac{ {M}^{2} d_{{\rm{R}}}^{2} }{4}  } \right| < \left| { \frac{{{{\rm{C}}_1}}}{{2{{\rm{C}}_2}}} - \frac{ {M}^{2} d_{{\rm{R}}}^{2} }{ {M}_{r}^{2} } } \right| $. The proof is complete. 
		
	\vspace{-0.2cm}
	\bibliographystyle{IEEEtran}
	\bibliography{IEEEabrv,ref_JointLocalization}   
\end{document}